\documentclass[lettersize,journal]{IEEEtran}
\usepackage{amsmath,amssymb,amsfonts,amsthm}
\usepackage{algorithm,algpseudocode}
\usepackage{array}
\usepackage[caption=false,font=normalsize,labelfont=sf,textfont=sf]{subfig}
\usepackage{textcomp}
\usepackage{stfloats}
\usepackage{url}
\usepackage{verbatim}
\usepackage{graphicx}
\usepackage{cite}
\usepackage{multirow}
\usepackage{xr-hyper}
\usepackage{hyperref}
\usepackage{enumitem}
\newcommand{\rr}{\mathop{{\rm I}\mskip-4.0mu{\rm R}}\nolimits}
\usepackage{tabularx}
\usepackage{float}

\usepackage{todonotes}

\newtheorem{theorem}{Theorem}
\newtheorem{lemma}{Lemma}
\newtheorem{definition}{Definition}
\newtheorem{proposition}{Proposition}
\newtheorem{property}{Property}

\newtheorem{assumption}{Assumption}

\newtheorem{problem}{Problem}
\newtheorem{remark}{\textbf{Remark}}

\newcommand{\Z}{\mathop{{\rm Z}\mskip-7.0mu{\rm Z}}\nolimits}

\DeclareMathOperator*{\argmin}{arg\,min}

\def\showfig{1}

\begin{document}


\title{A Feedback Linearized Model Predictive Control Strategy for Input-Constrained Self-Driving Cars}

\author{Cristian Tiriolo and Walter Lucia
\thanks{This work was supported in part by the Natural Sciences and Engineering Research Council of Canada (NSERC) and in part by the by the Fonds de recherche du Québec – Nature et Technologies (FRQNT), DOI https://doi.org/10.69777/346571.}%
\thanks{Cristian Tiriolo and Walter Lucia are with the Department of Cybersecurity and Intelligent Systems Engineering (CISE), Concordia University, Montreal, QC, H3G 1M8, CANADA, {\tt\small cristian.tiriolo@concordia.ca}, {\tt\small walter.lucia@concordia.ca}}
}



\maketitle

\begin{abstract}
This paper proposes a novel real-time affordable solution to the trajectory
tracking control problem for self-driving cars subject to longitudinal and steering angular velocity constraints. To this end, we develop a dual-mode Model Predictive Control (MPC) solution
starting from an input-output feedback linearized description of the vehicle kinematics. 
First, we derive the state-dependent input constraints acting on the linearized model and characterize their worst-case time-invariant inner approximation. Then, a dual-mode MPC is derived to be real-time affordable and ensuring, by design, constraints fulfillment,  recursive feasibility, and uniformly ultimate boundedness of the tracking error in an ad-hoc built robust control invariant region.  
%
The approach’s effectiveness and performance are experimentally validated via laboratory experiments on a  Quanser Qcar. The obtained results show
that the proposed solution is computationally affordable and with tracking capabilities that outperform two alternative control schemes. 
\end{abstract}

\section{INTRODUCTION}
The advent of self-driving cars marks a significant step forward in automotive technology, heralding a new era of road safety, traffic efficiency, and environmental sustainability. Central to the autonomous vehicle's operational integrity is trajectory tracking control, a critical aspect for ensuring the safety, reliability, and comfort of these vehicles \cite{dixit2018trajectory, stano2023model}. The term reference trajectory refers to a sequence of consecutive waypoints, with associated spatial and temporal information, that the vehicle is required to accurately track \cite{liu2021reinforcement}.  Achieving reliable trajectory tracking enables autonomous vehicles to navigate complex environments with precision, adaptability, and safety, thus accelerating the integration of autonomous technology into everyday transportation systems \cite{paden2016survey}. 

Extensive research has been conducted to develop control strategies to solve the trajectory tracking problem for autonomous cars \cite{li2021state}, ranging from simple non-model based solutions like the well-established PID controllers \cite{nie2018longitudinal,abatari2013using}, to more sophisticated nonlinear control solutions like sliding-mode controllers \cite{wu2019path,shirzadeh2019adaptive}, and adaptive backstepping control \cite{hu2021adaptive}. In the last decade, deep learning-based approaches have been applied to the control of autonomous vehicles due to their ability to self-optimize their behavior from data and adapt to complex and dynamic environments. \cite{kuutti2020survey,grigorescu2020survey} offer an exhaustive review of the most recent developments in the application of machine learning techniques to autonomous vehicle control. Despite their relatively high tracking performance, the biggest challenge pertaining to this class of algorithms remains their dependence on large, annotated datasets for training, which can be expensive and time-consuming to collect and maintain. Additionally, machine learning models often act as ``black boxes,'' offering limited interpretability regarding how decisions are made, which raises concerns about accountability and safety in critical applications like autonomous driving.
One common drawback of the above-discussed tracking strategies is their incapability to address input constraints, i.e., physical limitations of the computed control signal, which may lead to a lack of close-loop stability guarantees \cite{lima2018experimental}. 

On the other hand, Model Predictive Control (MPC) has emerged as a premier control strategy in this domain, owing to its ability to anticipate the future behavior of the vehicle and handle multiple constraints simultaneously \cite{camacho2007nonlinear,stano2023model}. In recent years, the application of MPC in autonomous vehicles has been extensively studied, highlighting its potential in managing the complex dynamics and uncertainties inherent to vehicular control \cite{stano2023model}. Notable works \cite{falcone2007predictive,borrelli2005mpc, frasch2013auto, batkovic2023experimental} have underscored the efficacy of MPC in navigating autonomous vehicles through dynamic environments. Despite these advancements, the deployment of MPC in real-world scenarios faces significant difficulties. Nonlinear MPC approaches \cite{pang2022practical,cho2023model}, while robust, are computationally demanding, posing challenges to their real-time implementation \cite{schwarting2017safe,nezami2023design}. Moreover, given the nonconvex nature of the MPC optimization, the solver algorithms may be characterized by uncertain convergence and suboptimality \cite{schwarting2017safe,eiras2021two}.  
Conversely, linear MPC techniques \cite{funke2016collision,beal2012model,nezami2023design} offer computational efficiency but at the expense of model fidelity. As a matter of fact, solutions that exploit error dynamics linearized around the reference trajectory have been found to be suboptimal~\cite{majd2019stable}.

When car-like vehicles are of interest, the existing literature shows that it is not trivial to design MPC frameworks that combine the real-time operational feasibility of linear formulations with the precision of nonlinear ones. Existing linear MPC strategies often resort to simplified models, sacrificing accuracy for computational speed. Meanwhile, the precision of nonlinear MPC comes at the cost of computational feasibility, limiting its practical application in autonomous vehicles. 
A possible way to mitigate the computational burdens of nonlinear approaches while preserving the accuracy of their model predictions is to resort to Feedback Linearization. This well-established technique recasts a nonlinear system into an equivalent linear one.
One of the first attempts to feedback-linearize the car-bicycle kinematics can be found in \cite{de2005feedback}, where both input-output FL and dynamic FL have been proposed. However, it can be shown that if such a linearized model is exploited for predictions, even simple box-like input constraints recast into state-dependent constraints leading to nonconvex MPC formulations \cite{deng2009input,simon2013nonlinear,tiriolo2022design}. \cite{deng2009input,simon2013nonlinear,tiriolo2022design}. 
Recently, the combination of MPC and FL has been used in \cite{DoProdan2024} to design a controller for a quadcopter, where the resulting time-varying constraints have been inner approximated. In \cite{tiriolo2022receding}, a similar idea has been used to solve a tracking problem for differential-drive robots. However, to the best of the author's knowledge, there are no existing  FL-based MPC solutions for input-constrained autonomous car-like vehicles.

\subsection{Paper's Contribution and Organization}\label{sec:paper_contributions}
In this paper, we design a dual-mode FL-based MPC strategy for self-driving cars capable of dealing with longitudinal and steering angular velocity constraints.
The contributions of the paper can be summarized as follows: 
\begin{itemize}
    \item it formally characterizes the worst-case time-invariant approximation of the parallelogram-shaped input constraint set acting on the input-output linearized car tracking error dynamics. Different from  \cite{tiriolo2022design}, the provided inner approximation does not require any assumption on the length of the parallelogram's sides.

    \item it provides the conditions to design a stabilizing terminal controller and associated control invariant region for the feedback linearized vehicle model; the proposed stabilizing controller exploits a feedforward term and the knowledge of the reference trajectory that mitigates the conservativeness of the control law proposed in \cite{tiriolo2022design}.
    \item It proposes a dual-mode linearized MPC scheme that ensures stable full-state tracking and input constraint fulfillment. Unlike existing nonlinear MPC formulations for car-like vehicles, the proposed approach requires solving a Quadratic Programming (QP) problem. Unlike linear MPC formulations, the proposed strategy does not introduce approximated model predictions along the used prediction horizon.
    \item It experimentally validates the proposed dual-mode MPC using laboratory experiments with a Quanser QCar and provides performance comparisons with the other two competing schemes. The developed code is available at the following GitHub repository: 
    \it \url{https://github.com/PreCyseGroup/Feedback-Linearized-MPC-for-self-driving-cars}.
\end{itemize}
%

%

The remainder of the paper is organized as follows. Section \ref{sec-preliminary} collects some preliminary definitions from control invariance theory and car's kinematic modeling. Moreover, it formally states the considered trajectory tracking problem. Section \ref{sec:prop-sol} describes the proposed dual-mode MPC strategies, with a formal proof of the obtained theoretical results. Section \ref{sec:exp-res} describes the experimental testbed, the performed experiments, and the obtained results. Finally, Section \ref{sec:conclusion} concludes the paper with some final remarks.

\section{Preliminaries and Problem Formulation}\label{sec-preliminary}

Given a matrix $M$ and a vector $v,$  $M{[i,:]}$ denotes the $i-th$ row of $M,$ $M{[i,j]}$ the $(i,j)$ entry of $M,$ and $v{[i]}$ the $i-th$ element of $v.$ A continuous-time function $f(t)$ is said of class $\mathcal{C}^3$ if $f(t)$ admits three continuous derivatives in its domain. 
Given $n>$ scalars  $m_i\in \rr$, $i=1,\dots,n,$ $M=diag(\left[m_1,\dots,m_n \right])\in \rr^{n\times n}$ defines a diagonal matrix with elements $m_i$ on the main diagonal. 
Given a variable $v,$ $v(k)$ denotes the values of $v$ at the
discrete sampling time instant $k\in \Z_+:=\{0,1,\ldots\}.$  
Given a discrete-time index $i\in \Z_+$ and the signal $v(k),$  $v(k+i|k)$ denotes the $i$-steps ahead prediction of $v$ from the time instant $k.$
\begin{definition} 
	\it{Consider a dynamical system 
$z(k+1)=f(z(k),u(k))
$   
    subject to the input constraints $u(k)\in\mathcal{U},\, \forall k\geq 0.$ A set $\Sigma$ is said to be control invariant if \cite{borrelli2017predictive}:
		$$
		\forall z(t)\in\Sigma, \, \exists u(k)\in\mathcal{U}\, :\,  f(z(k),u(k))\in\Sigma,\, \forall k\geq0 
		$$ 
	}
\end{definition}
\begin{definition} 
	\it{Consider a dynamical system $z(k+1)=f_a(z(k),w(k))$ subject to bounded disturbance $w(k)\in \mathcal{W},\,\forall\,k\geq 0.$ A set $\Sigma$ is said Robust Positive Invariant (RPI) if \cite{borrelli2017predictive}:
		$$
		\forall z(0)\in\Sigma,\, z(k)\in\Sigma,\,\forall w\in\mathcal{W}, \, \forall k>0 
		$$
 }
\end{definition}
\begin{definition}
    Consider a set $\mathcal{Q}$ neighborhood of the origin. The dynamical system $z(k+1)=f_a(z(k),w(k))$ is said to be Uniformly Ultimately Bounded (UUB) in $\mathcal{Q}$ if \cite{blanchini2008set} $\forall\mu>0,\, \exists T(\mu)>0$, such that $\forall \|z(0)\|\leq \mu,\, z(k)\in\mathcal{Q},\,\forall w\in\mathcal{W},\,\forall k\geq T(\mu)$.  
\end{definition}
\subsection{Car-like vehicle modeling}
Let's consider a rear-driven car-like vehicle whose continuous-time kinematics is described by \cite{de2005feedback}:
\begin{figure}[!h]
	\centering
	\includegraphics[width=0.55\linewidth]{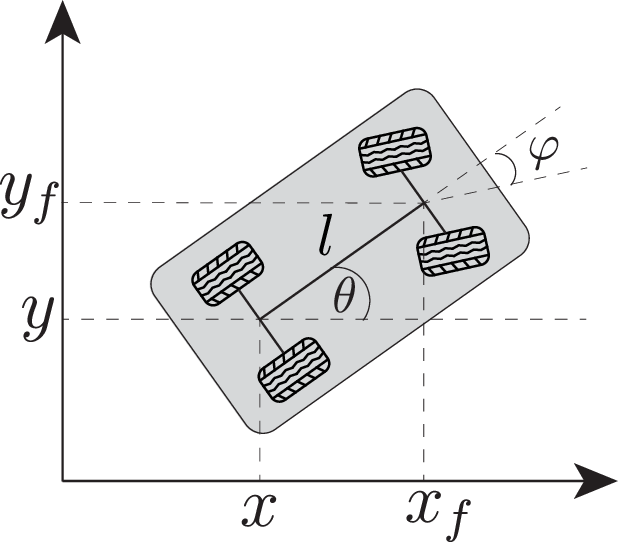}
	\caption{Car-like vehicle}
	\label{fig:car-model}
\end{figure}
\begin{equation}\label{eq:bicycle-model}
	\dot{q}(t)\!=\!\left[\begin{array}{c}
		\dot{x}(t)\\\dot{y}(t)\\\dot{\theta}(t)\\\dot{\varphi}(t)
	\end{array}\right]\!=\!\left[\begin{array}{c}
		\cos\theta(t) \\ \sin\theta(t)  \\ \frac{1}{l}\tan(\varphi(t)) \\ 0
	\end{array}\right]\!\!v(t)\!+\!\left[\begin{array}{c}
		0 \\ 0 \\ 0 \\ 1 
	\end{array}\right]\omega^s(t)
\end{equation}
\noindent
where $q=\left[x,y,\theta,\varphi \right]^T$ is the car's state, i.e., the
$x$ and $y$ are the Cartesian coordinates of the rear axis' midpoint, the heading angle $\theta$ of the vehicle, and the steering angle $\varphi.$
 On the other hand, $u=\left[v,\omega^s\right]^T$ are the control inputs, i.e., the longitudinal velocity of the vehicle and the steering angular velocity, respectively.
Moreover, $l$ represents the car's wheelbase, i.e,  the distance in meters between the front and rear wheels, and
$$
x_f=x+l\cos\theta \quad
y_f=y+l\sin\theta
$$
are the Cartesian coordinates of the front axis' midpoint. We assume the car's model to be subject to symmetrical box-like constraints representing the admissible longitudinal and steering angular velocities that the car can perform, i.e, 
\begin{equation}\label{eq:car-constraints}
	\begin{array}{c}
		u(t)\in \mathcal{U}_{car}:=\{u\in\rr^2:\, Tu\leq h\}
	\end{array}
\end{equation}
where 
$$
T=\left[\begin{array}{rr}
	-1 & 0 \\
	0 &  -1\\
	1 & 0 \\ 
	0 & 1 \\
\end{array}\right],\quad h=\left[\begin{array}{c}
	\overline{v} \\ \overline{\omega^s} \\ \overline{v}\\\overline{\omega^s}
\end{array}\right]
$$ and
$\overline{v},\overline{\omega^s}>0$ are given upper bounds.

By defining a sampling time $T_s>0$, and resorting to forward Euler discretization method, the following discrete-time kinematics is obtained:
\begin{equation}\label{eq:bicycle-discrete}
    \begin{cases}
        x(k+1)&=x(k)+T_sv(k)\cos(\theta(k))\\
        y(k+1)&=y(k)+T_sv(k)\sin(\theta(k))\\
        \theta(k+1)&=\theta(k)+T_s\frac{v(k)}{l}\tan(\varphi(k))\\
        \varphi(k+1)&=\varphi(k)+T_s \omega^s(k)
    \end{cases}
\end{equation}

In what follows, for the sake of simplicity, we refer to \eqref{eq:bicycle-discrete} using the following compact representation: 
$$
q(k+1)=f_{car}(q(k),u(k))
$$

\subsection{Problem's statement}
\begin{assumption}\label{ass:ref-trajectory}
	A sufficiently smooth and uniformly bounded reference trajectory is given by a high-level guidance algorithm in terms of the Cartesian coordinates $(x_r(t)$, $y_r(t))$ of the rear axis center of the car. By denoting with  $(\dot{x}_r(t)$, $\dot{y}_r(t)),$ $(\ddot{x}_r(t)$, $\ddot{y}_r(t))$ and $(\dddot{x}_r(t), \dddot{y}_r(t)),$  the reference velocities, accelerations and jerks, respectively,  the corresponding reference car's state is  $q_r(t)=\left[x_r(t),y_r(t),\theta_r(t),\varphi_r(t)\right]^T$, where $\theta_r(t)$ and $\varphi_r(t)$ are the heading and steering angles compatible with  the car-like model \eqref{eq:bicycle-model} and computed as \cite[Sec. 3.1]{de2005feedback}:

		\begin{equation}\label{eq:ref-traj-variable-state}
			\begin{array}{rcl} \theta_r(t)=&\text{ATAN}_2\left(\frac{\dot{y}_r(t)}{v_r(t)},\frac{\dot{x}_r(t)}{v_r(t)}\right)\\	\varphi_r(t)=&\arctan\left(\frac{l(\ddot{y}_r(t)\dot{x}_r(t)-\ddot{x}_r(t)\dot{y}_r(t))}{v_r(t)^3}\right)
			\end{array},
		\end{equation}
		%
	On the other hand, the corresponding feasible reference inputs signal is $u_r(t)=\left[v_r(t),\, \omega^s_r(t)\right]^T,$ where 
	\begin{equation}\label{eq:ref-traj-variable-inputs}
		\begin{array}{rcl} 	v_r(t)=&\sqrt{\dot{x}_r(t)^2+\dot{y_r}(t)^2}\\ 
			 \omega^s_r(t)=&lv_r\frac{\left(\dddot{y}_r\dot{x}_r-\dddot{x}_r\dot{y}_r\right)v_r^2-3\left(\ddot{y}_r\dot{x}_r-\ddot{x}_r\dot{y}_r\right)\left(\dot{x}_r\ddot{x}_r+\dot{y}_r\ddot{y}_r\right)}{v_r^6+l^2\left(\ddot{y}_r\dot{x}_r-\ddot{x}_r\dot{y}_r\right)^2}\\
		\end{array}
	\end{equation}
\end{assumption}

	\begin{remark}\label{remark:assumption_reference}
		Assumption~\ref{ass:ref-trajectory} requires that $q_r(t)$ must be compatible with the considered car-like dynamics. Such an assumption is leveraged by the proposed approach to later characterize the error dynamics as in  \eqref{eq:dt-error-model}. Consequently, reference trajectories $q_r(t)$ not compatible with the car-like dynamics are not of interest for this work. 
	\end{remark}

\begin{problem}\label{problem-formulation}
    Design a constrained state feedback controller 
    \begin{equation}
        u(t)=\phi(t,q(t),q_r(t),u_r(t))
    \end{equation}
    such that $u(t)\in\mathcal{U}_{car},\forall\,t\geq 0$ and stable full-state tracking is achieved, i.e.,  
	$$\exists \delta>0, \, t_0\geq0\, \text{ s.t. }\|\tilde{q}(t_0)\|< \delta \implies \|\tilde{q}(t)\|<\varepsilon, \,\forall t\geq t_0
	$$ 
 where $\tilde{q}(t)=q(t)-q_r(t)$ is the tracking error.
\end{problem}

\section{Proposed Solution}\label{sec:prop-sol}
In this section, the considered problem is addressed by combining feedback-linearization and MPC arguments. First, the control problem is described as a standard nonlinear MPC over a finite prediction horizon. Then, the nonconvex nature of the underlying MPC optimization is analyzed, and a novel predictive framework based on feedback linearization is proposed to recover a convex optimization problem that fulfills constraints while guaranteeing a bounded tracking error.
\subsection{Nonlinear MPC}\label{sec:NL-MPC}
Let's define
$\tilde{u}=u-u_r$ 
and a Linear-Quadratic (LQ) cost 
\begin{equation} \label{eq:cost_function}
	\begin{array}{lcr}
	J_N(k,\tilde{q}(k),\tilde{u}(k))=\displaystyle\sum_{i=0}^{N-1} &\!\!\!\!\!\!\tilde{q}(k+i+1|k)^TQ\tilde{q}(k+i+1|k)+\\&+\tilde{u}(k+i|k)^TR\tilde{u}(k+i|k)
	\end{array}
\end{equation}
where
$N>0$ is the prediction horizon, and $Q=Q^T\geq 0, \, Q\in\rr^n$, $R=R^T>0,\, R\in\rr^m$ are weighting matrices for the state and control input tracking errors, respectively. Then, the optimal control law that minimizes the defined cost function over the prediction horizon $N$ can be computed as:
\begin{subequations}\label{eq:NL-MPC-optimization}
    \begin{gather}
        u(k)=\displaystyle\argmin_{u(k),\dots,u(k+N-1)} J_N(k,\tilde{q}(k),\tilde{u}(k))\,\, s.t. \label{eq:NL-MPC-1}\\
         q(k+i+1|k)=f_{car}(q(k+i|k),u(k+i))\label{eq:NL-MPC-2}\\
         u(k+i)\in\mathcal{U}_{car}\label{eq:NL-MPC-3}\\
                  \Tilde{q}(k+N|k)\in\mathcal{Q}_N\label{eq:NL-MPC-4}\\
                  \quad i=0,1\dots N-1 \nonumber
    \end{gather}
\end{subequations}
where $\mathcal{Q}_N$ is a predefined set, PI with respect to an offline-designed feedback terminal control law $u_N(k)=\phi_N(k,q(k),q_r(k),u_r(k))\in\mathcal{U}_{car},\, \forall k\geq N$.
The above is known as dual-mode MPC, i.e., for the first $N$ steps,
the control law is obtained by solving the above optimization and applying the optimal solution in a receding horizon fashion, i.e., only the first sample $u(k)$ is applied to the system \eqref{eq:bicycle-model} and the optimization is solved at any sampling time. Then,  
once the error trajectory 
$\tilde{q}(k)$ reaches $\mathcal{Q}_N$, the control law $u_{N}(k)$ associated to $\mathcal{Q}_N$ is used.
%
\begin{remark}
    The above dual-mode MPC strategy guarantees stability and input constraint fulfillment, for any initial condition $\tilde{q}(0)$ such that the optimization problem \eqref{eq:NL-MPC-optimization} is feasible \cite{mayne2000constrained}. Consequently, under the effect of the dual-mode MPC control law, the tracking error is bounded with respect to any trajectory complying with Assumption \ref{ass:ref-trajectory}. 
\end{remark}
\begin{remark}\label{remark:terminal}
    Although appealing, optimization \eqref{eq:NL-MPC-optimization} suffers from the following drawbacks:
    \begin{itemize}       
        \item The optimization problem is  nonconvex due to the presence of the constraints \eqref{eq:NL-MPC-2} and \eqref{eq:NL-MPC-4}. Moreover, the obtained solutions may suffer from local minima problems \cite{eiras2021two};
                \item The computational burden associate to \eqref{eq:NL-MPC-optimization}, especially for large prediction horizon $N$, may not allow the real-time implementation of the control scheme;                 
        \item The computation of $\mathcal{Q}_N$ and associated state-feedback controller $u_N$ is not trivial for the nonlinear vehicle kinematic model \eqref{eq:bicycle-model}.
    \end{itemize}

\end{remark}

Motivated by the above drawbacks, in what follows a novel MPC formulation based on feedback linearization arguments is proposed. 
In particular, first, the input-output linearization proposed in \cite{de2005feedback} is exploited to obtain a linear description of the car's kinematics. Then, inspired by the idea introduced in \cite{tiriolo2022design}, the time-varying input constraints acting on the linearized model and their worst-case realization are analytically characterized. Finally, the obtained constrained model and worst-case arguments are used to design a tracking control strategy that ensures stability, recursive feasibility, and input constraint fulfillment.


\subsection{Input-Output Feedback Linearization}\label{sec:input-output-lin}

Here, the input-output feedback-linearization introduced in \cite{de2005feedback} is used to obtain a linearized description of the car's kinematic model. 

Let's define two new outputs
\begin{equation}\label{eq:FL-output-transf}
	z=\left[\begin{array}{c}
		z_1 \\z_2
	\end{array}\right]=	\left[\begin{array}{c}
		x+l\cos(\theta)+\Delta\cos(\theta+\varphi)\\
		y+l\sin(\theta)+\Delta\sin(\theta+\varphi)
	\end{array}\right]
\end{equation}
representing the Cartesian position of a point $P$ at a distance $\Delta>0$ from the center of front wheels' axis, and a new state vector $\eta=\left[\eta_1,\,\eta_2\right]^T=\left[\theta,\,\varphi\right]^T$. 
Then, by resorting to the following input transformation depending on $\eta$:
\begin{equation}\label{eq:FL-input-transform}
	\begin{array}{c}
		w=M(\eta)u,\\
		M(\eta)=
		\left[\!\!\begin{array}{cc}
			\cos(\eta_1)-\tan(\eta_2)(\sin(\eta_1)+\frac{\Delta}{l}s_1) & -\Delta s_1 \\			\sin(\eta_1)+\tan(\eta_2)(\cos(\eta_1)+\frac{\Delta}{l} c_1) & \Delta c_1
		\end{array}\!\!\right]
	\end{array}
\end{equation}
where $
s_1 = \sin(\eta_1+\eta_2)\quad c_1 = \cos(\eta_1+\eta_2),
$
 the kinematic model \eqref{eq:bicycle-model} is recast into
\begin{subequations}
	\label{eq:feedback_linearized_model}
	\begin{gather}
		\dot{z}=w \label{eq:FL-linearized-model}\\
		\Dot{\eta}=O(\eta)w\label{eq:FL-internal-dynamics}
	\end{gather}
\end{subequations}
where
\begin{equation}\label{eq:O}
O(\eta)=
\left[
\begin{array}{cc}
		\frac{\sin(\eta_2) c_1}{l} &    \frac{\sin(\eta_2) s_1 }{l}\\
  \frac{-\sin (\eta_2) c_1}{l}-\frac{s_1}{\Delta} & \frac{-\sin (\eta_2)s_1}{l}+\frac{c_1}{\Delta}
\end{array}
\right] 
\end{equation}
Notice that \eqref{eq:FL-linearized-model} defines a two-single-integrator model subject to a decoupled nonlinear internal dynamics
\eqref{eq:FL-internal-dynamics}. 
The above decoupled system can be discretized via forward Euler discretization method, obtaining:
\begin{subequations}
	\label{eq:dt-feedback_linearized_model}
	\begin{gather}
		z(k+1)=Az(k)+Bw(k),\quad A=I_{2\times2},\quad B=T_sI_{2\times2}\label{eq:dt-FL-linearized-model}\\
		\eta(k+1)=\eta(k)+T_sO(\eta(k))w(k)\label{eq:dt-FL-internal-dynamics}
	\end{gather}
\end{subequations}
\begin{property}\label{property1}
    The input output feedback-linearization \eqref{eq:FL-output-transf}-\eqref{eq:feedback_linearized_model} and forward Euler discretization method commute for the car's kinematic model \eqref{eq:bicycle-model}.
\end{property}
\begin{proof}
Let's consider the output transformation \eqref{eq:FL-output-transf} and its first-order derivative
	$$
	\begin{array}{rcl}
	    \dot{z}_1&=&\dot{x}-l\sin(\theta)\dot{\theta}-\Delta\sin(\theta+\varphi)(\dot{\theta}+\dot{\varphi})\\
		\dot{z}_2&=&\dot{y}+l\cos(\theta)\dot{\theta}+\Delta\cos(\theta+\varphi)(\dot{\theta}+\dot{\varphi})
	\end{array}
	$$
Then, under forward Euler discretization arguments one obtains
	\begin{equation}
     { \scriptsize
	        \begin{split}
	    \frac{z_1(k+1)-z_1(k)}{T_s}&=\frac{x(k+1)-x(k)}{T_s}-l\sin(\theta)\frac{\theta(k+1)-\theta(k)}{T_s}-\\&- \Delta\sin(\theta+\varphi)\left[\frac{\theta(k+1)-\theta(k)}{T_s}+\frac{\varphi(k+1)-\varphi(k)}{T_s}\right]\\
		\frac{z_2(k+1)-z_2(k)}{T_s}&=\frac{y(k+1)-y(k)}{T_s}+l\cos(\theta)\frac{\theta(k+1)-\theta(k)}{T_s}+\\&+\Delta\cos(\theta+\varphi)\left[\frac{\theta(k+1)-\theta(k)}{T_s}+\frac{\varphi(k+1)-\varphi(k)}{T_s}\right]
	\end{split}
        }
	\end{equation}
By substituting  $x(k+1)$, $y(k+1)$, $\theta(k+1).$  $\varphi(k+1)$ with the right-hand sides of  \eqref{eq:bicycle-discrete} and rewriting the equation in a compact form, the resulting discrete-time evolution of $z_1$ and $z_2$ is:
\begin{equation}\label{eq:proof_commute}
	 z(k+1)=z(k)+T_s M(\eta(k))u(k)
  \end{equation}
Finally, by using the input transformation~\eqref{eq:FL-input-transform}, the discrete-time feedback linearized system \eqref{eq:dt-feedback_linearized_model} is obtained,
which is equal to the discrete-time system obtained by discretization of \eqref{eq:FL-linearized-model}.
Similarly, the nonlinear internal dynamics \eqref{eq:FL-internal-dynamics} can be discretized obtaining \eqref{eq:dt-FL-internal-dynamics}.
Hence, input-output linearization and discretization commute.
\end{proof}


\subsection{Tracking Error Model and Input Constraint Characterization}
Here, the feedback-linearized tracking error model is formally derived.  By applying the output transformation \eqref{eq:FL-output-transf}, the reference output for the input-output linearized system \eqref{eq:feedback_linearized_model} is given by 
\begin{equation}\label{eq:lin-ref-traj}
	z_r=
	\left[\begin{array}{c}
		x_r+lcos(\theta_r)+\Delta\cos(\theta_r+\varphi_r)\\
		y_r+lsin(\theta_r)+\Delta\sin(\theta_r+\varphi_r)
	\end{array}\right]
\end{equation}
Similarly, reference inputs for \eqref{eq:feedback_linearized_model} can be computed via \eqref{eq:FL-input-transform}, obtaining
\begin{equation}\label{eq:w_r}
    w_r(t)=M(\theta_r(t),\varphi_r(t))u_r(t)
\end{equation}
By defining the error vectors $\tilde{z}=z-z_r$ and $\tilde{w}=w-w_r$, $\Tilde{\eta}=\eta-\eta_r,\, \eta_r=\left[\theta_r,\phi_r\right]^T$, the input-output linearized and internal tracking error dynamics are given by:
\begin{subequations}\label{eq:error-lin-sys}
    \begin{gather}
        \dot{\tilde{z}}(t)=\tilde{w}(t)\label{eq:linearized-tracking-error-dyn}\\
      \dot{\tilde{\eta}}(t)=\kappa(\Tilde{\eta},\Tilde{w},\eta_r,w_r,t)=O(\eta(t))w(t)-O(\eta_r(t))w_r(t)\label{eq:internal-tracking-error-dyn}
    \end{gather}
\end{subequations}
which can be discretized by resorting to the Euler forward method and re-written as
\begin{subequations}\label{eq:dt-error-model}
	\begin{gather}
		\tilde{z}(k+1)=A\tilde{z}(k)+Bw(k)-Bw_r(k), \label{eq:dt-error-lin-sys}\\ 
		A=I_{2\times 2},\quad B=T_sI_{2\times 2}\\
         \tilde{\eta}(k+1)=\tilde{\eta}(k)+T_s\kappa(\Tilde{\eta}(k),\tilde{w}(k),\eta_r(k),w_r(k),k)\label{eq:dt-internal-tracking-error-dyn}
	\end{gather}
\end{subequations}
\begin{remark}\label{remark:disturbance_ball_bound}
Since the reference trajectory is assumed to be bounded, then also $w_r(k)$ is bounded and the set of admissible $w_r(k)$ can be over-approximated by a ball $\mathcal{W}_r\subset \rr^2$ of radius $r_d,$ i.e., 
\begin{equation}\label{eq:worst-case-disturbance}
	w_r \in \mathcal{W}_r=\{w_r\in\rr^2:w_r^TW_r^{-1} w_r\leq 1\},\,\,W_r=r_d^2I_{2\times2}
\end{equation}
\end{remark}
Given the feedback linearized tracking error dynamics, the following lemma establishes sufficient conditions for bounded internal dynamics. 
\begin{lemma}\label{lem:full-state-track}

	\textit{If the reference trajectory $q_r(t)$  and reference inputs $v_r(t)$ and $\omega^s_r(t)$ complies with Assumption~\ref{ass:ref-trajectory},
		then  the tracking-error zero dynamics $\dot{\tilde{\eta}}=\kappa(\eta,0,\eta_r,w_r,t)$
		are asymptotically stable \cite[Theorems 1-3]{wang2003full}. 
		Consequently, if \eqref{eq:linearized-tracking-error-dyn}  is stabilized by a particular control law $\bar{\omega}(k)$ defined $\forall\,k\geq 0$, then, stable  full-state tracking is achieved under the control law  
		$\bar{u}(k)=M^{-1}(\eta(k)){\bar{w}}(k),\, \forall\, k\geq 0$ for \eqref{eq:bicycle-discrete}			 \cite{wang2003full}.\hfill $\Box$
	}
\end{lemma}
By applying the transformation \eqref{eq:FL-input-transform} to the input constraints \eqref{eq:car-constraints}, the tracking-error dynamics \eqref{eq:dt-error-model} are subject to  the following time-varying polyhedral input constraints, depending on the internal dynamics state $\eta$ i.e.,

\begin{equation}\label{eq:time-varying-input-constraint}
	w\!\in\! \mathcal{U}(\eta)\!=\!\{w\in\rr^2\!:  L(\eta)w\leq h\}, \, L(\eta)=TM^{-1}(\eta)
\end{equation}

The following lemma analytically characterizes the polyhedral set  $\mathcal{U}(\eta),$ which rotates and resizes as a function of $\eta.$ 
\if\showfig 1
\begin{figure}[!h]
	\centering
	\includegraphics[width=0.63\linewidth]{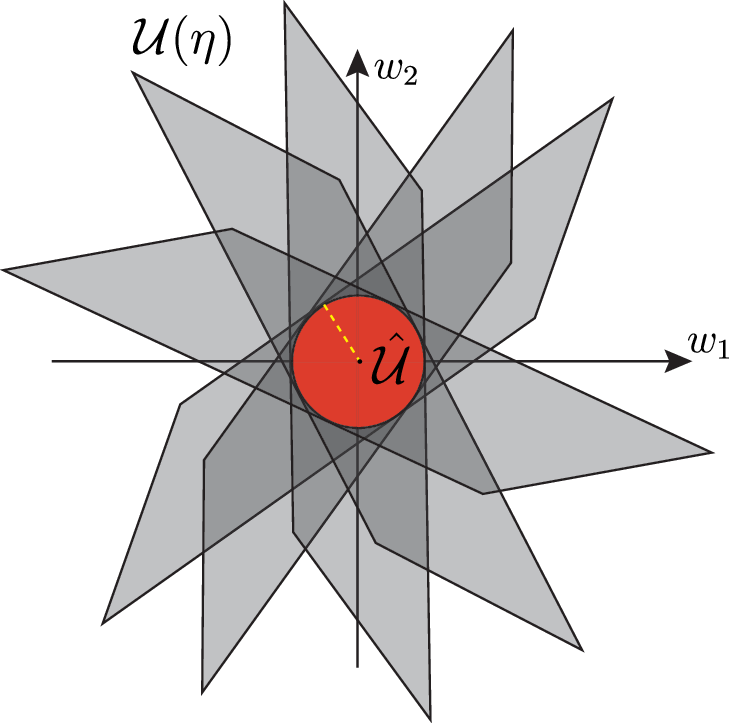}
	\caption{Time-varying input constraint set and its worst-case approximation}
	\label{fig:input-set}
\end{figure}
\fi
\begin{lemma}\label{lemma:input_constraint_set}
    The polyhedral input constraint set \eqref{eq:time-varying-input-constraint} is a time-varying parallelogram that admits the following worst-case circular inner approximation (see Fig. \ref{fig:input-set}):
    \begin{equation}\label{eq:worst-case-constr}
		\begin{array}{lcr}
  \hat{\mathcal{U}}=\displaystyle\bigcap_{\forall\eta} \mathcal{U}(\eta)=
  \{w\in\rr^2|w^Tw\leq \hat{r}^2\},\\\hat{r}=\min \left(\frac{\Delta l \overline{ \omega^s}}{\sqrt{\Delta^2+l^2}}, \overline{v}\right)
		\end{array}
    \end{equation}
\end{lemma}
\if\showfig 1
\begin{figure}[!h]
	\centering
	\includegraphics[width=1\linewidth]{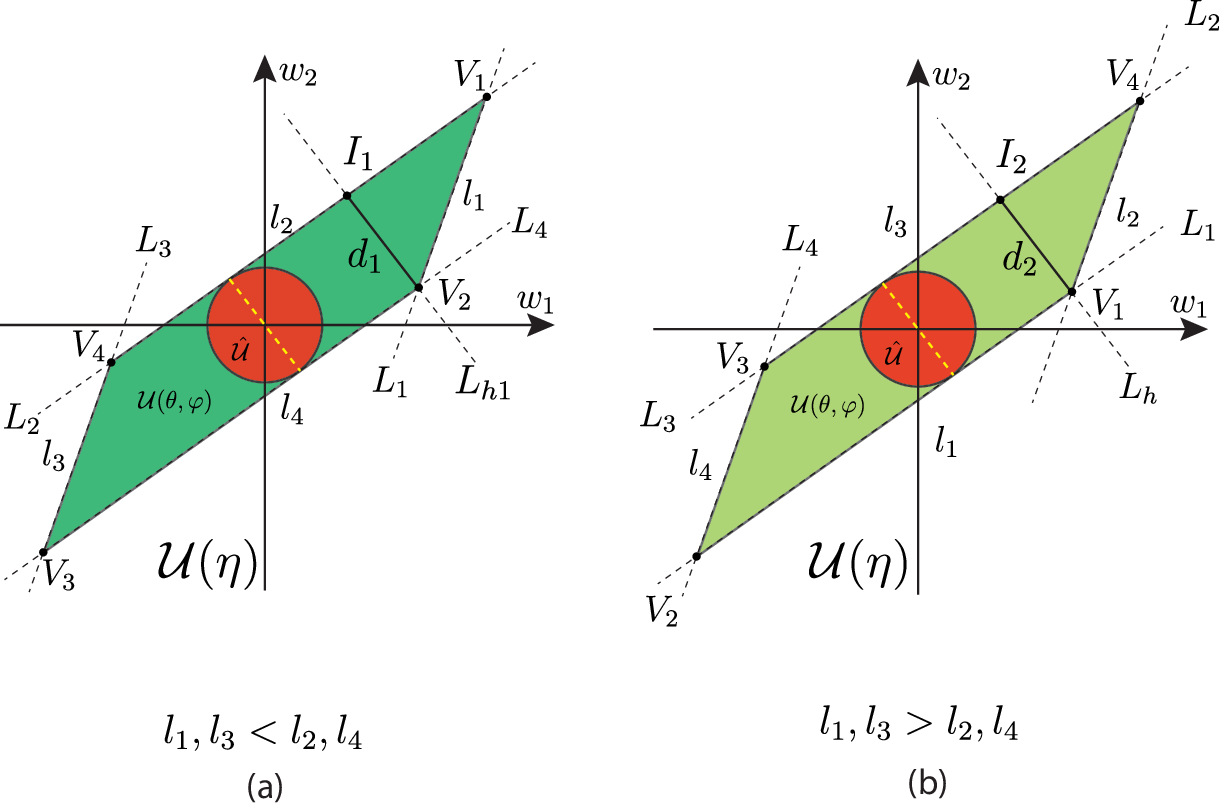}
	\caption{Possible side length configurations for $\mathcal{U}(\eta)$}
	\label{fig:polyhedron-configuration}
\end{figure}
\fi
\begin{proof}
     By defining  $
s_2 = \sin(\eta_1+2\eta_2)$ and $c_2 = \cos(\eta_1+2\eta_2),
$
the shaping matrix of the polyhedral set $\mathcal{U}(\eta)$ can be re-written as: 
    %
    \begin{equation*}
		\begin{array}{c}
			L (\eta)\!=\!\!
			\left[\!\!	\begin{array}{cc}
				-\frac{\cos(\eta_1)+ c_2 }{2} & -\frac{\sin(\eta_1)+ s_1 }{2} \\ 	\frac{\Delta s_1 -\Delta\sin(\eta_1)+2l s_2 }{2\Delta l}&\frac{-\Delta c_2 +\Delta\cos(\eta_1)-2l c_1 }{2\Delta l} \\
				\frac{\cos(\eta_1)+ c_2 }{2} &\scriptstyle \frac{\sin(\eta_1)+ s_1 }{2} \\ -\frac{\Delta s_1 -\Delta\sin(\eta_1)+2l s_2 }{2\Delta l}&-\frac{-\Delta c_2 +\Delta\cos(\eta_1)-2l c_1 }{2\Delta l}
			\end{array}\!\!\right]
		\end{array}
	\end{equation*}
    By intersecting the four hyperplanes, the four vertices have the following analytical expression \cite{tiriolo2022design}:
    \begin{subequations}\label{eq:poly-vertices}
   \begin{gather}
		 \!\!\!\!\!V_1(\eta)=\left[\begin{array}{c}
			\frac{\Delta s_1 l\overline{\omega^s} + \Delta\sin(\eta_1)l\overline{\omega^s} + \Delta \overline{v}\cos(\eta_1) - \Delta \overline{v} c_2  - 2 c_1 l\overline{v}}{(2l\cos(\eta_2))}\\
			\frac{-\overline{\omega^s}\Delta l \cos(\eta_1) - \overline{\omega^s} \Delta l  c_2  - \Delta  s_1  \overline{v} + \Delta\sin(\eta_1)\overline{v} - 2 s_2 l \overline{v}}{2l\cos(\eta_2)}
		\end{array}\right] \label{eq:poly-vertices-1}\\
	 \!\!\!\!\!	V_2(\eta)=\left[\begin{array}{c}
			\frac{-\Delta s_1 l \overline{\omega^s} - \Delta\sin(\eta_1)l\overline{\omega^s} + \Delta \overline{v}\cos(\eta_1) - \Delta \overline{v} c_2  - 2 c_1 l\overline{v}}{2l\cos(\eta_2)}\\
			\frac{\overline{\omega^s}\Delta l \cos(\eta_1)+\overline{\omega^s}\Delta l  c_2  - \Delta s_1  \overline{v}+ \Delta\sin(\eta_1)\overline{v} -2 s_2  l \overline{v}}{2l \cos(\eta_2)}
		\end{array}\right]\label{eq:poly-vertices-2}\\
	V_1(\eta)=-V_3(\eta),\quad V_2(\eta)=-V_4(\eta) \label{eq:poly-vertices-3}
    \end{gather}
	\end{subequations}
    By computing the Cartesian distances between the vertices, each side of the parallelogram has the following length:
    \begin{equation}
		\begin{array}{c}
			l_1=l_3=2\sqrt{\Delta^2 \overline{\omega^s}^2}\\
			l_2(\eta_2)=l_4(\eta_2)=2\sqrt{\frac{v_1^2(-\Delta^2\cos(2\eta_2) + \Delta^2 + 2l^2)}{l^2(\cos(2\eta_2) + 1)}}
		\end{array}
    \end{equation}	
    It can also be noted that the angular coefficients 
    of the four lines $L_1,\,L_2,\,L_3,\,L_4$ defining the polyhedron, 
    namely 
$m_1(\eta_1,\eta_2),\, m_2(\eta_1,\eta_2), m_3(\eta_1,\eta_2),\, m_4(\eta_1,\eta_2),$ 
are such that $m_1(\eta_1,\eta_2)=m_3(\eta_1,\eta_2)$, and $m_2(\eta_1,\eta_2)=m_4(\eta_1,\eta_2)$, $\forall \eta_1\in\rr,\eta_2\in\left[-\overline{\varphi},\overline{\varphi}\right]$. Consequently, $\mathcal{U}(\theta)$ is a time-varying parallelogram whose side lengths depend on the state variable $\eta_2$, and on the car's parameters $\Delta,\, l, \, \overline{v},\, \overline{\omega^s}$.

    In order to find the radius of the smallest circle inscribed in the polyhedral set, we resort to geometric arguments. By referring to  Fig.~\ref{fig:polyhedron-configuration}, two different cases must be considered 
    \textit{(a)} $l_1,l_3<l_2,l_4$ and \textit{(b)} $l_1,l_3>l_2,l_4$.
    Depending on the specific case, the diameter of the inscribed circular set can be found either as the distance of vertex $V_2$ from the point $I_{1}$, or the distance of vertex $V_1$ from the point $I_{2}$. Note that $I_1$ is the intersection of the line $L_2$ and the orthogonal to $L_2$ crossing $V_2$ (case \textit{(a)}), whereas $I_2$ is the intersection of the line $L_3$ and the orthogonal to $L_3$ crossing $V_1$ (case \textit{(b)}).
    
    Formally, the lines  $L_2$ and $L_3$  crossing $(V_1,\,V_4)$ and $(V_4,\,V_3)$, respectively are 
	\begin{equation}
    	\begin{array}{c}
			L_2: \quad w_2=m_2(\eta_1,\eta_2)w_1+h_2(\eta_1,\eta_2) \\ m_2(\eta_1,\eta_2)=-\frac{L\left[2,1\right]}{L\left[2,2\right]},\quad  h_2(\eta_1,\eta_2)=\frac{h\left[2\right]}{L\left[2,2\right]}
		\end{array}
	\end{equation}
    \begin{equation}
    \begin{array}{c}
			L_3: \quad w_2=m_3(\eta_1,\eta_2)w_1+h_3(\eta_1,\eta_2) \\ m_3(\eta_1,\eta_2)=-\frac{L\left[3,1\right]}{L\left[3,2\right]},\quad h_3(\eta_1,\eta_2)=\frac{h\left[3\right]}{L\left[3,2\right]}
		\end{array}
	\end{equation}
\end{proof}
Moreover, by resorting to simple geometric arguments, the equations of the lines $L_{h1}$ and $L_{h2}$ are:
\begin{equation}
		L_{h1}: \quad w_2-V_2\left[2\right]=-\frac{1}{m_2(\eta_1,\eta_2)}(w_1-V_2\left[1\right])
\end{equation}
\begin{equation}
		L_{h2}: \quad w_2-V_1\left[2\right]=-\frac{1}{m_3(\eta_1,\eta_2)}(w_1-V_1\left[1\right])
\end{equation}
Then, the points $I_{1}$ and $I_2$ can be computed intersecting $L_2$ with $L_{h1}$ and $L_3$ with $L_{h2}$, obtaining
\begin{equation}
		I_1=\left[\begin{array}{cc}
			-m_2(\eta_1,\eta_2) & 1 \\ \frac{1}{m_2(\eta_1,\eta_2)} & 1 
		\end{array}\right]^{-1}\left[\begin{array}{c}
			h_2(\eta_1,\eta_2) \\ \frac{V_2\left[1\right]}{m_2(\eta_1,\eta_2)}+V_2\left[2\right]
		\end{array}\right]
\end{equation}
\begin{equation}
		I_2=\left[\begin{array}{cc}
			-m_3(\eta_1,\eta_2) & 1 \\ \frac{1}{m_3(\eta_1,\eta_2)} & 1 
		\end{array}\right]^{-1}\left[\begin{array}{c}
			h_3(\eta_1,\eta_2) \\ \frac{V_1\left[1\right]}{m_3(\eta_1,\eta_2)}+V_1\left[2\right]
		\end{array}\right]
\end{equation}
By noticing that the inscribed circles have diameters equal to $d_1=\overline{V_2I_1}$ (case \textit{(a)}) and  $d_2=\overline{V_1I_2}$ (case \textit{(b)}), the radii, namely $r_1(\eta_2)$ and $r_2(\eta_2),$ are
%
%
\begin{equation}
		\begin{array}{c}
		r_1(\eta_2)=\frac{1}{2}d=\frac{1}{2}\sqrt{(I_1-V_2)^T(I_1-V_2)}=\\\vspace{-0.3cm}\\
		=\frac{\Delta l \overline{\omega^s}}{\sqrt{\Delta^2-\Delta^2 \cos(\eta_2)^2 + l^2}}
		\end{array}
\end{equation}
\begin{equation}
		\begin{array}{c}
		r_2(\eta_2)=\frac{1}{2}d=\frac{1}{2}\sqrt{(I_2-V_1)^T(I_2-V_1)}=\\\vspace{-0.3cm}\\
		=\sqrt{\frac{\overline{v}^2}{\cos(\eta_2)^2}}
		\end{array}
\end{equation}
Furthermore, since  $\eta_2\in\left[-\overline{\eta_2},\overline{\eta_2}\right],\overline{\eta_2}<\frac{\pi}{2},$ the minimum value of $r_1(\eta_2)$ and $r_2(\eta_2)$ is obtained for $\eta_2=0,$ that it is equals to
 $$r_1=r_1(0)=\frac{\Delta l \overline{\omega^s}}{\sqrt{\Delta^2+l^2}},\quad r_2=r_2(0)=\overline{v}$$
Consequently, \eqref{eq:worst-case-constr} defines the worst-case circle inscribed in $\mathcal{U}(\eta),\forall\eta,$ concluding the proof.

\subsection{Robust Invariant Control Design}\label{sec:RPI-region}
By using similar arguments to the ones exploited in \cite{tiriolo2022design}, the linearized tracking error dynamics can be exploited to design a state feedback controller that fulfills the prescribed time-varying and state-dependent input constraints in a properly defined robust invariant region.
%

\begin{proposition}\label{prop:terminal-controller}
    The circular set
    \begin{equation}\label{eq:terminal-region-shaping}
        \Sigma_N=\{\tilde{z}\in\rr^2 | \tilde{z}^TS\tilde{z}\leq 1\},\quad S=\frac{1}{\hat{r}^2}K^TK
    \end{equation}
    is RPI for \eqref{eq:dt-error-lin-sys} under the  state-feedback controller
    \begin{equation}\label{eq:terminal-controller}
        w(k)=K\tilde{z}(k)+\hat{w}_r(k)
    \end{equation}
   where $\hat{w}_r(k)$ is the optimal solution of the following Quadratic Programming (QP) problem:
    \begin{subequations}\label{eq:terminal-optimization}
			\begin{gather}
		\hat{w}_r(k)=\displaystyle\arg\min_{\hat{w}_r} \|\hat{w}_r-w_r(k)\|_2^2 \quad s.t.\\
				K\tilde{z}(k)+\hat{w}_r\in \mathcal{U}(\eta)\label{eq:RCI-optim-constr}	
			\end{gather}
    \end{subequations}
   if  $K$ is such that
    \begin{equation}\label{eq:feasibility-constraint}
	\lambda^{-1}A_{cl}^TS^{-1}A_{cl}+(1-\lambda)^{-1}B^TW_{r_d}^{-1}B\leq S^{-1}.
    \end{equation}
    where $A_{cl}=A-BK$, $\lambda=1-\sqrt{\xi}$ and $\xi$ is the only repeated eigenvalue of the matrix $G^TB^TW_{r_d}^{-1}BG$, with $G$ such that $G^TS^{-1}G=I_{2\times 2}.$ 
\end{proposition}
\begin{proof}

For the disturbance-free model (i.e. obtained from \eqref{eq:dt-error-lin-sys} when $w_r(k)=0,\forall\,k$), any stabilizing controller $w(k)=K\Tilde{z}(k)$ fulfills the input constraint for any $\Tilde{z} \in  \Sigma_N$ where $\Sigma_N$ is as in \eqref{eq:terminal-region-shaping}. Specifically, the set  $\Sigma_N$ is obtained by plugging the state-feedback controller $w(k)=K\Tilde{z}(k)$ into the circular region \eqref{eq:worst-case-constr}, representing the worst-case set of admissible input for \eqref{eq:dt-error-model}. Moreover, since in the disturbance-free case, \eqref{eq:feasibility-constraint} reduces to a standard Lyapunov inequality $(A-BK)^TS^{-1}(A-BK)-S^{-1}\leq 0,$ then if $K$ fulfills \eqref{eq:feasibility-constraint} then $\Sigma_N$ is also a positively invariant region. 

On the other hand, in the presence of $w_r(k)\neq 0$ and under the control law $w(k)=K\tilde{z}(k)+\hat{w}_r(k),$ the closed-loop system is
    \begin{equation}\label{eq:closed-loop-sys}
        \begin{array}{rcl}
             \Tilde{z}(k+1)&=&(A-BK)\Tilde{z}(k)+B(\hat{w}_r(k)-w_r(k))\\
             &=&A_{cl}\Tilde{z}(k)+Bw_d(k)
        \end{array}
    \end{equation}
with $w_d=\hat{w}_r(k)-w_r(k).$
If $\hat{w}_r(k)$ is given by the solution of \eqref{eq:terminal-optimization}, then the control law 
$w(k)=K\tilde{z}(k)+\hat{w}_r(k)$ fulfils the input constraints for any $\Tilde{z} \in  \Sigma_N$. Moreover, $w_d$ is bounded inside the set $\mathcal{W}_r$ (with the worst-case happening when $\hat{w}_r(k)=0$). Finally, as proven in \cite[Section 3]{kolmanovsky1998theory}, 
if $K$ fulfils \eqref{eq:feasibility-constraint}, then
$ \Sigma_N$ is RPI for \eqref{eq:closed-loop-sys}, concluding the proof.
\end{proof}

\begin{remark}\label{remark:design}
Note that Proposition~\ref{prop:terminal-controller} provides sufficient  conditions under which the circular set $\Sigma_N$ is 
RPI under a stabilizing gain $K.$ As long as the actual computation of $K$ and $S$ is concerned, any existing iterative or joint approach can be used. For example,  \cite{tiriolo2022design,kolmanovsky1998theory} resorts to an iteratively procedure, while \cite{fang2004analysis} define a single LMI optimization to jointly design $K$ and $S.$ 
\end{remark}

\subsection{Feedback Linearized Model Predictive Control}
Under input-output linearization arguments, optimization \eqref{eq:NL-MPC-optimization} can be equivalently rewritten as follows:
\begin{subequations}\label{eq:exact-FL-optimization}
	\begin{gather}
		\displaystyle\min_{w(k),\dots,w(k+N-1)}\displaystyle  J_N(k,\left[\tilde{z}^T(k),\tilde{\eta}^T(k)\right]^T,\tilde{w}(k)) \label{eq:exact-FL-1}\\
		\tilde{z}(k+i+1|k)=A\tilde{z}(k+i|k)+Bw(i)-Bw_r(i)\label{eq:exact-FL-2}\\
        \tilde{\eta}(k+i+1|k)=\tilde{\eta}(k+i|k)+\\+T_s\kappa(\Tilde{\eta}(k+i|k),\tilde{w}(k+i),\eta_r(k+i),w_r(k+i),k+i)\label{eq:exact-FL-3}\\
        L( \tilde{\eta}(k+i|k)+\eta_r(k+i))w(k+i)\leq h  \label{eq:exact-FL-4}\\ \forall i=0,1,\dots N-1\nonumber\\
        \tilde{z}^T(k+N|k)S\tilde{z}(K+N|k)\leq 1 \label{eq:exact-FL-5}
	\end{gather}
\end{subequations}
however, the resulting optimization problem \eqref{eq:exact-FL-optimization} is still nonconvex due to constraints \eqref{eq:exact-FL-3}-\eqref{eq:exact-FL-4}  which require nonlinear predictions of the internal nonlinear dynamics $\forall i\geq1.$ One possible way to convexify the optimization problem (and make the resulting MPC problem convex and real-time affordable) is to substitute the polyhedral constraint \eqref{eq:exact-FL-4} with its quadratic worst-case approximation \eqref{eq:worst-case-constr} $\forall i\geq 1$, which is independent of the nonlinear dynamics state $\eta$. Moreover, since $\eta(k)$ can be measured, the conservativeness of the MPC controller can be mitigated using actual input constraint set \eqref{eq:time-varying-input-constraint} for $i=0.$ Consequently, optimization \eqref{eq:exact-FL-optimization} can be rewritten as: 
\begin{subequations}\label{eq:conservative-FL-optimization}
	\begin{gather}
		\displaystyle\min_{w(k),\dots,w(k+N-1)}\displaystyle J_N(k,\tilde{z}(k),\Tilde{w}(k)) \label{eq:conservative-FL-1}\\
		\tilde{z}(k+i+1|k)=A\tilde{z}(k+i|k)+Bw(i)-Bw_r(i)\label{eq:conservative-FL-2}\\
        \forall i=0,1,\dots N-1\nonumber\\
       L(\eta(k))w(k)\leq h \label{eq:conservative-FL-4}\\ 
        w(k+i)^Tw(k+i)\leq \hat{r}^2,\, i=1,\dots N-1 \label{eq:conservative-FL-5}\\
           \tilde{z}^T(k+N|k)S\tilde{z}(K+N|k)\leq 1 \label{eq:conservative-FL-6}
	\end{gather}
\end{subequations}
which is a Quadratically Constrained Quadratic Programming (QCQP) problem. 

\begin{proposition}\label{proposition:QCQPimplementation}
 The QCQP problem \eqref{eq:conservative-FL-optimization} can be rewritten in the following standard form:
\begin{subequations}\label{eq:MPC-compact-optim}
    \begin{gather}
       \mathbf{w^*}=\displaystyle\arg \min_{\mathbf{w}}\displaystyle \frac{1}{2}\mathbf{w}^TH\mathbf{w}+p^T\mathbf{w}\,s.t.\label{eq:MPC-compact-1}\\
        \hat{L}(\eta(k))\mathbf{w}\leq h \label{eq:MPC-compact-2}\\
        \mathbf{w}^T\hat{Q}_u\mathbf{w}\leq 1 \label{eq:MPC-compact-3}\\
        \mathbf{w}^T\Theta^T_NS\Theta_N\mathbf{w}\!+\!2\Tilde{z}^T(k)\Psi^T_NS\Theta_N\mathbf{w}\leq 1\!-\!\Tilde{z}^T(k)\Psi_N^TS\Psi_N\Tilde{z}(k)\label{eq:MPC-compact-4}
    \end{gather}
\end{subequations}
where $$H=\Theta^T\hat{Q}\Theta+\hat{R}$$ and $$p=\Theta^T\hat{Q}\Psi\Tilde{z}(k)-\Theta^T\hat{Q}\Theta\mathbf{w_r}-\hat{R}\mathbf{w_r}$$
with
\begin{equation*}
	\Psi=\left[\begin{array}{c}
		A \\ A^2 \\ \vdots \\ A^N
	\end{array}\right], \quad \displaystyle\begin{array}{lcr}
	\Theta=\left[\begin{array}{cccc}
		B&0&\dots&0 \\
		AB&B&\dots&0\\
		\vdots&\vdots&\ddots&\vdots\\
		A^{N-1}B&A^{N-2}B&\dots&B
	\end{array}\right]\end{array}
\end{equation*}
\begin{equation*}
    \Psi_N=A^N,\quad \Theta_N=\left[A^{N-1}B,\,A^{N-2}B,\dots,AB,\,B\right]
\end{equation*}
\begin{equation*}
		\hat{Q}=\left[\begin{array}{cccc}
		Q & 0 & \dots & 0 \\ 
		0 & Q & \dots & 0 \\
		\vdots & \vdots &\ddots& \vdots\\
		0 & 0 & \dots & Q  
	\end{array}\right],
    \quad \hat{R}=\left[\begin{array}{cccc}
		R & 0 & \dots & 0 \\ 
		0 & R & \dots & 0 \\
		\vdots & \vdots &\ddots& \vdots\\
		0 & 0 & \dots & R  
	\end{array}\right]
\end{equation*}
\begin{equation*}
\hat{Q}_u=\left[\begin{array}{cccc}
		0 & 0 & \dots & 0 \\ 
		0 & \frac{1}{\hat{r}^2}I_{2\times 2} & \dots & 0 \\
		\vdots & \vdots &\ddots& \vdots\\
		0 & 0 & \dots & \frac{1}{\hat{r}^2}I_{2\times 2}
	\end{array}\right]
\end{equation*}
$$
\hat{L}(\eta(k))=\left[L(\eta(k)),\,0,\dots 0\right]
$$
\end{proposition}

\begin{proof}
    Let's define the decision variables vector $\mathbf{w}=\left[w(k),\, w(k+1),\dots w(k+N-1)\right]^T$, and the predicted reference input vector $\mathbf{w_r}=\left[w_r(k),\, w_r(k+1),\dots w_r(k+N-1\right]^T$. Then, using \eqref{eq:conservative-FL-2}, the model predictions $\mathbf{\Tilde{z}}=\left[\Tilde{z}(k+1|k),\dots, \Tilde{z}(k+N|k)\right]^T$ can be rewritten in a compact form as $\mathbf{\Tilde{z}}=\Psi\Tilde{z}(k)+\Theta(\mathbf{w}-\mathbf{w}_r)$. As a consequence the cost function
    $$
    \begin{array}{lcr}
         J_N(k,\tilde{z}(k),\tilde{w}(k))=\frac{1}{2}\displaystyle\sum_{i=0}^{N-1} \tilde{z}(k+i+1|k)^TQ\tilde{z}(k+i+1|k)\\+\tilde{w}(k+i|k)^TR\tilde{w}(k+i|k)      
    \end{array}
   $$ 
    can be rewritten as $J(\mathbf{w})=\frac{1}{2}[\left(\Phi\Tilde{z}(k)+\Theta(\mathbf{w}-\mathbf{w_r})\right)^T\hat{Q}\left(\Phi\Tilde{z}(k)+\Theta(\mathbf{w}-\mathbf{w_r})\right)+(\mathbf{w}-\mathbf{w_r})^T\hat{R}(\mathbf{w}-\mathbf{w_r})]=\frac{1}{2}\mathbf{w}^TH\mathbf{w}+p^T\mathbf{w}+c$. Notice that in optimization \eqref{eq:MPC-compact-optim} the term $c$ has been dropped since it does not affect the optimal solution of the optimization.  By applying the same arguments, it is easy to show that $  L(\eta(k+i|k))w(k+i)\leq h,\, i=0,1,\dots N-1 \iff \hat{L}(\eta(k))\mathbf{w}\leq h$, $ w(k+i)^Tw(k+i)\leq \hat{r}^2,\, i=1,\dots N-1 \iff \mathbf{w}^T\hat{Q}_u\mathbf{w}\leq 1 $, and $\tilde{z}^T(k+N|k)S\tilde{z}(k+N|k)\leq 1\iff \mathbf{w}^T\Theta^T_NS\Theta_N\mathbf{w}+2\Tilde{z}(k)\Psi^T_NS\Theta_N\mathbf{w}\leq 1-\Tilde{z}^T(k)\Psi_N^TS\Psi_N\Tilde{z}(k) $
\end{proof}
\begin{remark}\label{rem:QP-formulation}
    The QCQP problem can be recast into a computationally more affordable QP problem. Specifically, the quadratic constraints \eqref{eq:MPC-compact-3} and \eqref{eq:MPC-compact-4} can be replaced with polyhedral inner approximations. In particular, by defining two polyhedral sets $\mathcal{P}_w=\{w\in\rr^2: P_ww\leq p_w\}\subset\hat{\mathcal{U}},\, P_w\in\rr^{n_w\times 2},\, p_w\in\rr^{n_w}$ and $\mathcal{P}_N=\{\tilde{z}\in\rr^2: P_{\Tilde{z}_N}\Tilde{z}\leq p_{\Tilde{z}_N}\}\subset\Sigma_N,\, P_{\Tilde{z}_N}\in\rr^{n_N\times 2},\, p_{\Tilde{z}_N}\in\rr^{n_{\Tilde{z}_N}}$, where $n_w$ and $n_{\Tilde{z}_N}$ are the number of sides of the polyhedral approximations $\mathcal{P}_w$ and $\mathcal{P}_{\Tilde{z}_N}$, respectively. Then, constraint \eqref{eq:MPC-compact-3} can be replaced with 
    \begin{equation}\label{eq:approx-poly-constr1}
        \hat{P}_w\mathbf{w}\leq \hat{p}_w
    \end{equation}
    where 
    \begin{equation*}
        \hat{P}_w=\left[\begin{array}{cccc}
		0 & 0 & \dots & 0 \\ 
		0 & P_w & \dots & 0 \\
		\vdots & \vdots &\ddots& \vdots\\
		0 & 0 & \dots & P_w  
	\end{array}\right], \quad \hat{p}_u=\left[\begin{array}{c}
		0 \\ 
		p_w \\
	    \vdots\\
		p_w
	\end{array}\right]
    \end{equation*}
    Similarly, the quadratic constraint \eqref{eq:conservative-FL-6} can be replaced with its polyhedral approximation $P_{\Tilde{z}_N}\tilde{z}(k+N)\leq p_{\Tilde{z}_N}$ which can be rewritten as a function of the decision variable $\mathbf{w}$, i.e.,
    \begin{equation}\label{eq:approx-poly-constr2}
         P_{\Tilde{z}_N}\Theta_N\mathbf{w}\leq p_{\Tilde{z}_N}-P_{\Tilde{z}_N}\Psi_N\Tilde{z}(k)
    \end{equation}
    Therefore, replacing \eqref{eq:MPC-compact-3}-\eqref{eq:MPC-compact-4} with \eqref{eq:approx-poly-constr1}-\eqref{eq:approx-poly-constr2}, a QP optimization is obtained. Notice that the conservativeness  and computational complexity of the derived QP problem depends on the number of sides of the polyhedra approximations, i.e. $n_w$ and $n_N$, which are additional design parameters. 
\end{remark}

All the above developments can be collected into the computable Algorithm~\ref{alg:FL-MPC}, which, as proved in the following theorem, provides a solution to  Problem~\ref{problem-formulation}.
\begin{algorithm}[!h]
	\caption{Dual-Mode Feedback-Linearized Tracking MPC (Dual-Mode FL-MPC) algorithm}\label{alg:FL-MPC}
	\noindent	\textit{Offline:}
	\begin{algorithmic}[1]
		\State Find $K$ and $S$ to fulfill the condition \eqref{eq:feasibility-constraint} 
        \State Set the prediction horizon $N$ such that \eqref{eq:MPC-compact-optim} is feasible with the initial condition $\tilde{z}(0)=z(0)-z_r(0)$
	\end{algorithmic}
	\textit{Online:}
	\begin{algorithmic}[1]
		\State  Estimate $x(k),\,y(k), \,\theta(k),\, \varphi(k)$ and compute $\Tilde{z}(k)=z(k)-z_r(k)$, and $\eta(k)=\left[\theta(k),\varphi(k)\right]^T$. \label{alg-step-1}
		\State  Compute $L(\eta(k))$ as in \eqref{eq:time-varying-input-constraint} and $w_r(k)$ as in \eqref{eq:w_r};
		\If{$\tilde{z}(k)\notin\Sigma_N$}
		\State Find $\mathbf{w}^*$ solving \eqref{eq:MPC-compact-optim} and set $w(k)={w}^*(k)$
		\Else{} 
         \begin{equation}\label{eq:terminal-law}
		    w(k)=K\tilde{z}(k)+\hat{w}_r(k)
		\end{equation} 
        where $\hat{w}_r(k)$ is the optimal solution of \eqref{eq:terminal-optimization}
		\EndIf
		\State  Compute 					\vspace{-0.2cm}
		\begin{equation} \label{eq:final_control_law}
			\left[v(k),\omega^s(k)\right]^T=M^{-1}(\eta(k)){w}(k)
		\end{equation}
		and apply it to the car;  $k\leftarrow k+1$, go to \ref{alg-step-1};
	\end{algorithmic}
\end{algorithm}
\begin{theorem}\label{theorem} \it
For any $\tilde{z}(0)$ such that \eqref{eq:MPC-compact-optim} is feasible, the tracking FL-MPC strategy described in Algorithm~\ref{alg:FL-MPC} provides a solution to Problem~\ref{problem-formulation}.
\end{theorem} 
\begin{proof}

The proof can be divided into two  parts: 

{\it (I) Stability and input constraint fulfillment of the linearized tracking error dynamics}:
First, let's consider the input-output linearized model \eqref{eq:dt-feedback_linearized_model}. If at the generic time $k$, \eqref{eq:MPC-compact-optim} admits a solution for a given initial condition $\Tilde{z}(k)$ and for some $N>0$, then the optimal control sequence $\{w^*(k),\,w^*(k+1),\dots,\, w^*(k+N-1)\}$ with $w^*(k+i)\in\mathcal{U}(\eta(k)),\forall \eta(k),\, \forall i=0,1,\dots N-1$, is such that $\Tilde{z}(k+N)\in\Sigma_N$. At time $k+1$, a feasible solution to optimization \eqref{eq:MPC-compact-optim} can be constructed from the optimal solution at time $k$, i.e., $\{w^*(k+1),w^*(k+2)\dots,\, w^*(k+N-1),\, K\Tilde{z}(k+N) \}$. 

Indeed, the last control move $K\Tilde{z}(k+N)$ is, by construction, always  constraint-admissible inside the RPI region $\Sigma_N.$ 
As a consequence, the optimization \eqref{eq:MPC-compact-optim} is recursively feasible ensuring that, in at most $N$ steps, $\tilde{z}(k)$ is steered into $\Sigma_N.$
Then, given the RPI nature of $\Sigma_N,$ we can also conclude that 
$\Tilde{z}(k)$ is UUB in $\Sigma_N.$ Furthermore, since the used input-output linearization and discretization commutes (see property \eqref{property1}), the linearized tracking error dynamics \eqref{eq:error-lin-sys} is stable under the effect of the proposed dual-mode MPC.

{\it (II) Bounded Tracking Error for \eqref{eq:bicycle-model}}:
As proven in part (I), $w(k)$ computed by Algorithm~\ref{alg:FL-MPC}  stabilizes the feedback linearized error dynamics \eqref{eq:dt-error-lin-sys}. Therefore, given the result of Lemma~\ref{lem:full-state-track} and by applying the input transformation~\eqref{eq:FL-input-transform}, the control law \eqref{eq:final_control_law} solves the considered reference tracking problem with a bounded tracking error $\tilde{q}(k).$ 
\end{proof}

\begin{remark}\label{remark:asymptotic-convergence}
	Theorem~1 ensures full-state bounded tracking error, which does not ensure asymptotic convergence of the physical coordinates $(x,y)$ to $(x_r,y_r)$. This limitation is intrinsic to feedback linearization of nonholonomic systems, where stable output tracking does not, in general, guarantee asymptotic full-state tracking \cite{wang2003full}. Nevertheless, under the stated assumptions on the reference trajectory and the invariance of the terminal set, the tracking error is guaranteed to remain bounded in a small neighbourhood of the origin.
\end{remark}

\begin{remark}\label{remark:disturbances}
If the considered car-like model is subject to additive bounded disturbances acting on the first two state components, then it is still possible to use the proposed control scheme with some adaptations. First, in this setup, under the used input-output feedback linearization,  the car-like model can be decoupled as in \eqref{eq:feedback_linearized_model} where an additive bounded disturbance acts only on the linear component \cite{henson1997nonlinear}. Consequently,  given the result stated in Lemma~\ref{lem:full-state-track}, the design objective in Problem~\ref{problem-formulation} is met if it is possible to design a controller capable of stabilizing the resulting linearized system subject to bounded disturbance. A few modifications are required to design such a controller using the proposed approach. In particular, an additional bounded disturbance term appears in the discretized error dynamic model \eqref{eq:dt-error-lin-sys} used for control design purposes. Now, for the offline computation of the RPI region $ \Sigma_N$ in Sec.~\ref{sec:RPI-region}, it is possible to combine  (with a Minkosky set sum) the newly introduced disturbance with the one `$-B\omega_r(k)$'  related to the reference trajectory $r(k).$   On the other hand, for the online computation of the control input, it is necessary to switch from the MPC formulation in \eqref{eq:conservative-FL-optimization} to a robust counterpart where the new unknown but bounded disturbance terms can be taken care using a standard Robut MPC methodology where the disturbance is handled by means of proper Minkowski set difference operations on the set of constraint \cite{chisci2001systems}. Nevertheless, such an approach presents a certain degree of conservatism, and it is viable as long as the introduced additional disturbance term does not lead to empty sets when Minkowski set sum and difference operations are used.
\end{remark}

 	

\section{Experimental Results}\label{sec:exp-res}
In this section, the experimental results, obtained using a Quanser Qcar\footnote{\url{https://www.quanser.com/products/qcar/}}, are presented to show the effectiveness of the proposed FL-MPC tracking controller and compare it with
 alternative solutions. In particular, the chosen competitors are the nonlinear MPC tracking strategy described in Section \ref{sec:NL-MPC} (namely ``Nonlinear MPC''), and the the constrained adaptive backstepping controller developed in \cite{hu2021adaptive} (namely ``Backstepping''). Two different versions of the proposed FL-MPC scheme have been tested. The first one
 exactly follows Algorithm~\ref{alg:FL-MPC} (hereafter referred to as ``Dual-mode FL-MPC''). On the other hand, the second one (namely ``FL-MPC'') executes Algorithm~\ref{alg:FL-MPC} but it never activates the terminal control law, i.e., it solves \eqref{eq:MPC-compact-optim} for any $k\geq0.$    
\if\showfig 1
\begin{figure}[!h]
	\centering
	\includegraphics[width=0.85\linewidth]{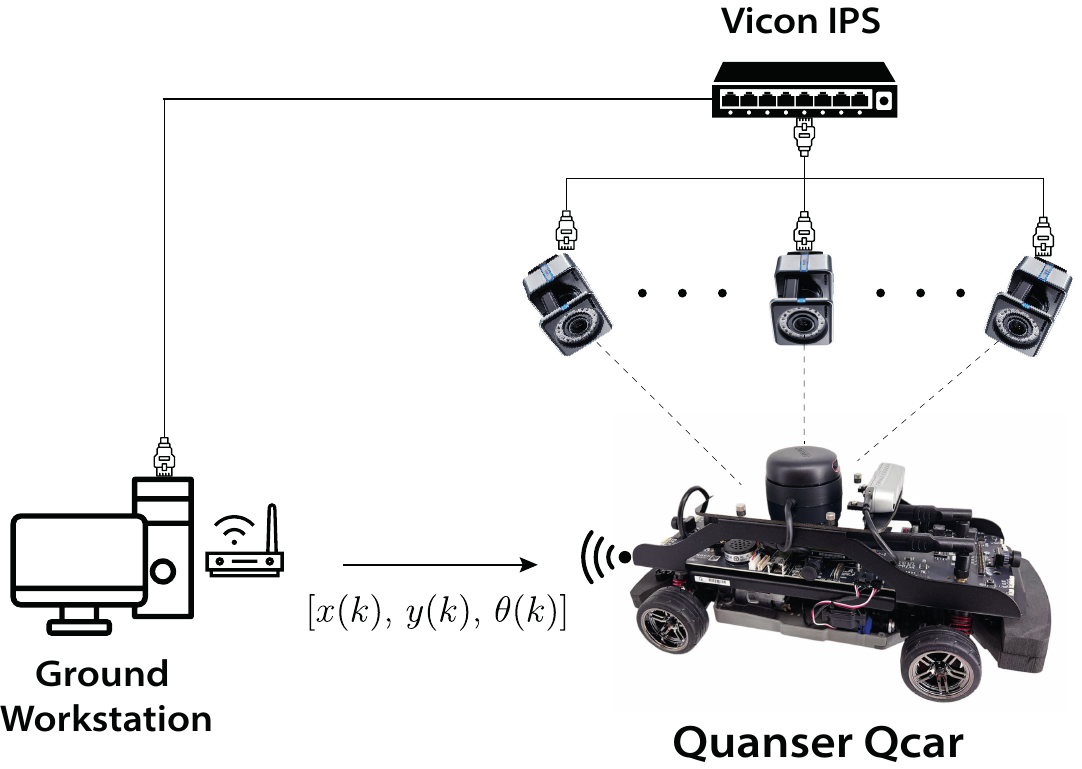}
	\caption{Proposed experimental setup}
	\label{fig:setup}
\end{figure}
\fi
\subsection{Experimental Setup}\label{sec:setup-exp}
The considered experimental setup is depicted in Fig. \ref{fig:setup}, and it consists of:
\begin{enumerate}[label=\alph*)]
	\item a Quanser Qcar;
	\item a camera-based Indoor Positioning System (IPS);
	\item a ground workstation;
	\item a Wifi communication channel between the ground workstation and the car.
\end{enumerate}
The autonomous car-like vehicle is the Quanser Qcar open-architecture prototype, which is designed for academic research experiments.  The car has a size of $0.39\,\times \, 0.19 \,\times \, 0.20\, m$, weights $2.7\, kg$, and its wheelbase measures $l=0.256 m$. Onboard, the car is equipped with different sensors (encoders gyroscope, accelerometer, magnetometer, lidar, and depth and RGB cameras)
 an \textit{NVIDIA® Jetson™ TX2 with 2 GHz quad-core ARM Cortex-A57 64-bit + 2 GHz Dual-Core NVIDIA Denver2 64-bit CPU and 8GB memory}. The considered maximum  longitudinal speed is $\overline{v}=1\, m / s$, while the maximum steering angular velocity is $\overline{\omega^s}=10\,rad / s$. In addition, due to the vehicle's mechanical construction, the front wheel's steering angle cannot exceed  $\overline{\varphi}=0.6\,rad$. The developed tracking algorithms have been implemented and cross-compiled in \textit{C} language and run onboard on the Nvidia Jetson CPU.  A sampling time $T_s=0.01 \,s$ has been used for all the performed tests. 
 
 The IPS consists of a set of 12 Vicon Vero cameras connected via a wired connection to the ground workstation. The camera system is used to localize the car in the workspace, similar to a standard GPS system. In particular, the cameras detect and track a set of reflective markers placed on the Qcar. The positions of the markers are real-time collected and processed on the ground workstation by the Vicon Tracker software, which  accurately reconstructs, via a triangulation algorithm, the position and orientation of car.

 The ground workstation is a desktop computer consisting of a \textit{13th Gen Intel(R) Core(TM) i9-13900KF} CPU, a \textit{NVIDIA GeForce RTX 4070} GPU, and 64GB of RAM. The workstation is used to estimate the pose of the car and broadcast it via a TCP/IP communication channel.
\subsubsection{Configuration of the proposed controller}\label{sec:exp-res-prop-contr} 
To implement the proposed tracking controller strategy, the following parameters have been considered: $\Delta=0.35$, $Q=I_{2\times 2}$, $R=0.01\cdot I_{2\times 2}$, $K=5I_{2\times 2}$.  These parameters have been experimentally tuned to achieve good tracking performance and avoid numerical issues resulting from the inversion operation performed in \eqref{eq:final_control_law}. 
The radius of the worst-case circular input constraint set has been computed as in \eqref{eq:worst-case-constr}, obtaining $\hat{r}=1$. 
The reference trajectory is built to comply with \eqref{eq:worst-case-disturbance}, with $r_d=11.54$. As a consequence, the feedback control gain $K=4I_{2\times 2}$ has been chosen such that RPI condition \eqref{eq:feasibility-constraint} is satisfied, with $A_{cl}=0.96I_{2\times 2}$, $S=16 I_{2\times 2}$,  $G=0.25 I_{2\times2}$, $\xi=4.69\cdot 10^{-8}$, $\lambda=0.9998$, computed as outlined in in Proposition \ref{prop:terminal-controller}. The idea described in Remark~\ref{rem:QP-formulation} has been used to obtain a QP formulation of the derived QCQP optimization  \eqref{eq:MPC-compact-optim}. 
In particular, the quadratic constraints \eqref{eq:MPC-compact-3}-\eqref{eq:MPC-compact-4} have been inner approximated using  two decahedra, 
i.e., two polyhedral sets defined as in  \eqref{eq:approx-poly-constr1}-\eqref{eq:approx-poly-constr2}, with $n_w=n_N=10$.
The control optimization problem has been solved on the car's processing unit, considering a prediction horizon $N=10$ using an \textit{Active Set} solver algorithm. The \textit{Active Set} algorithm uses a \textit{Cholesky} decomposition of the Hessian matrix $H$, which, being constant for the proposed optimization, has been precomputed to reduce the online computational load. The computational times obtained for the used solver are reported in Section \ref{sec:comp-time}. As a disclaimer, we acknowledge that our implementation is not optimized and that more efficient solvers would reduce computation time significantly. In fact, adopting specialized solvers like qpOASES, HPIPM, or acados could further extend the scalability of the compared approaches, e.g., enabling longer horizons or tighter real-time deadlines.

\subsubsection{Car's state estimation} 
 The state of the car $q(k)$ is onboard estimated by means of an  Unscented Kalman Filter (UFK).
%
%
In particular, the implemented UKF 
is outlined in \cite{van2001square} and it exploits the nonlinear kinematic model \eqref{eq:bicycle-discrete} and different collected sensor information: (i) the estimated position and orientation of the Car provided by the Vicon Camera System, (ii) encoder and (iii) IMU (gyroscope and accelerometer) measurements.
The UKF has been configured with the following parameters: process and measurement covariance matrices $Q_{UKF}=diag([10^{-3},\, 10^{-3},\,10^{-1},\,10])$, $R_{UKF}=diag([2\cdot10^{-5},\, 2\cdot10^{-5},\,10^{-4},\,10^{-5}])$, initial state estimation covariance matrix $P^0_{UKF}=diag([10^{-5},\, 10^{-5},\, 10^{-6},\,10^{-6}])$, sigma-points parameters $\alpha=0.9$, $\beta=2$, $\kappa=0$ (the interested reader shall refer to \cite{van2001square} for further details about the used parameters).

\subsubsection{Reference trajectory generation}
A reference trajectory complying with assumption \ref{ass:ref-trajectory} and the car maximum steering angle $\overline{\varphi}=0.6\,rad,$ has been generated using a cubic spline interpolation method.
The interpolation algorithm receives in input a sequence of waypoints describing the desired path and, in output, it assigns a crossing time based on path curvature and desired average speed. 
Then, each waypoint is interpolated using quintic splines, obtaining the position $x_r,\,y_r$, velocity $\dot{x}_r,\,\dot{y}_r$, acceleration $\ddot{x}_r,\,\ddot{y}_r$, and jerk $\dddot{x}_r,\,\dddot{y}_r$, needed to compute the reference car's state $q_r$, and inputs $u_r$, as in \eqref{eq:ref-traj-variable-state}-\eqref{eq:ref-traj-variable-inputs}.

\subsubsection{Configuration of the competitor schemes}\label{sec:exp-res-comp-contr} 
Each competitor scheme has been configured to obtain the best tracking performance in the performed experiments. Specifically, the nonlinear MPC optimization \eqref{eq:NL-MPC-optimization} has been implemented without terminal control law,  considering the LQ cost matrices, $Q=diag(\left[1,\, 1,\, 0.1,\, 0.1\right])$, $R=diag(\left[0.01,\,0.01\right])$, and a prediction horizon $N=5$. The nonlinear optimization has been solved using \textit{fmincon} solver configured with \textit{Sequential Quadratic Programming (SQP)} method. 
On the other hand, the backstepping algorithm developed in \cite{hu2021adaptive} has been tuned using $d_x=0.2$, $\sigma=0.3$, $\alpha_c=0.4$, $k_3=2.$ Such  parameters have been experimentally determined to provide the best tracking performance in the considered setup, which is slightly different from the one in \cite{hu2021adaptive} in terms of reference trajectory, state estimation algorithm and wheel actuation mechanism.
The interested reader shall refer to \cite{hu2021adaptive} for a detailed explanation of the used parameters . It is worth mentioning that the steering command generated by the algorithm is subject to undesired chattering effects, typical of backstepping control algorithms \cite{tanner2003backstepping}. In order to mitigate such an undesired effect, the control signal has been prefiltered using a low-pass filter with cut-off frequency $\omega_c=100 \frac{rad}{s}$.

\subsubsection{Evaluation of the tracking performance}
\label{sec:comp-time} To evaluate the tracking performance of the proposed controller and alternative schemes, the Integral Square Error (ISE) ($\int_{0}^{T_f}e(t)^2dt$) and Integral Time Squared Error (ITSE) $(\int_{0}^{T_f}t e(t)^2dt)$  indexes have been used. In particular, for each performed experiment, three different error signals are measured: path distance error $e_{xy}(k)=\|\left[x(k),\,y(k)\right]-\left[x_r(k),\,y_r(k)\right]\|_2$, heading angle error $e_\theta(k)=\theta(k)-\theta_r(k))$, and steering angle error $e_\varphi(k)=\varphi(k)-\varphi_r(k))$. Then, for each collected error signal the ISE and ITSE indexes are computed, i.e., $ISE_{x,y}$, $ITSE_{x,y}$, $ISE_{\theta}$, $ITSE_{\theta}$,  $ISE_{\varphi}$, $ITSE_{\varphi}$. 

\subsubsection{Evaluation of computational times} 
	To assess the computation complexity of the approaches  Dual FL-MPC and FL-MPC and Nonlinear MPC, the computational times required by each algorithm have been measured. 
In the performed analysis, only the controller computation has been considered, i.e., the state-estimation algorithm as well as the sensor processing have been neglected. 

\subsection{Results}
\begin{figure}[!h]
	\centering
	\includegraphics[width=1\linewidth]{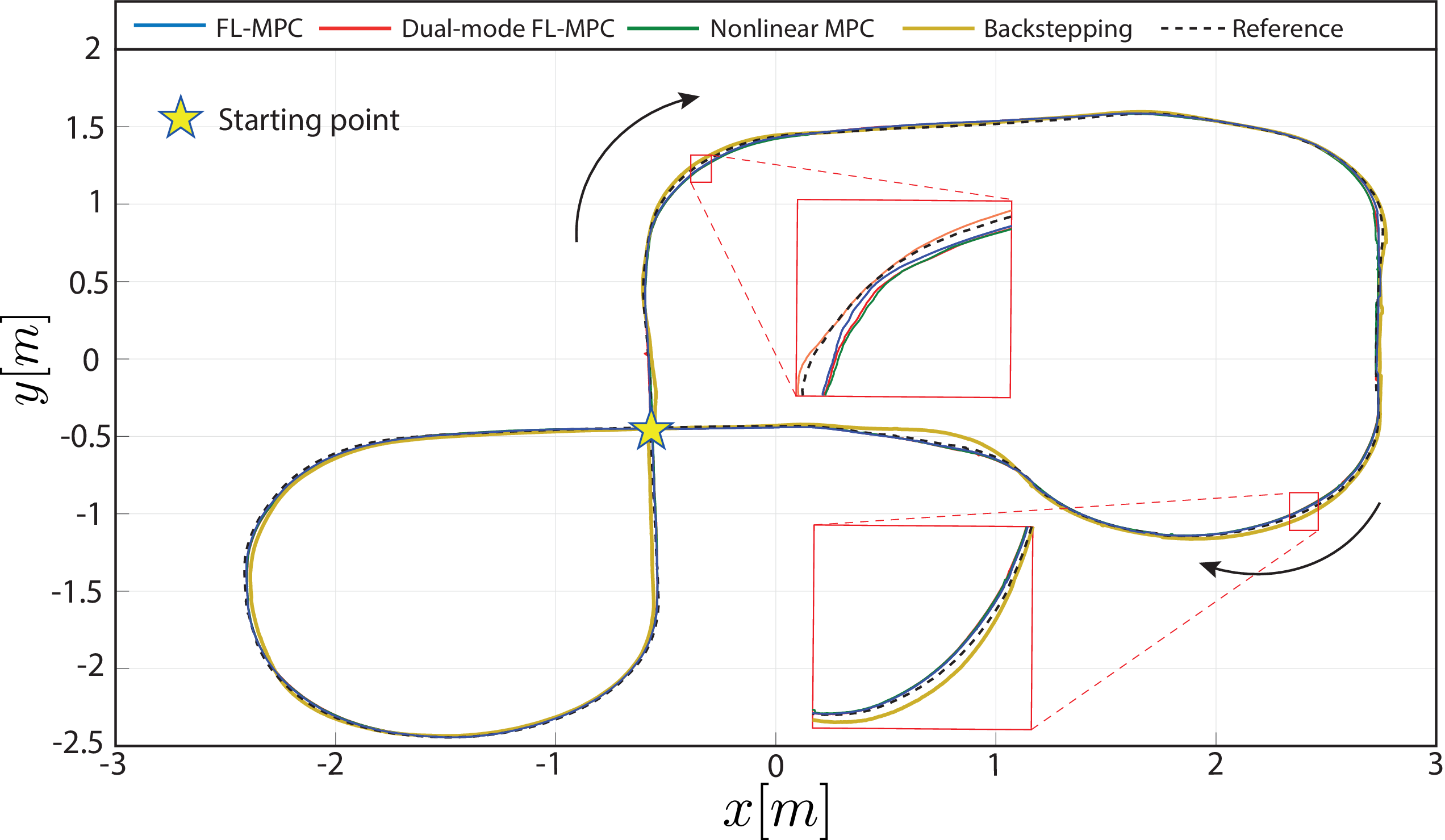}
	\caption{Experimental results: Trajectory}
	\label{fig:results-traj}
\end{figure}
\begin{figure}[!h]
	\centering
	\includegraphics[width=1\linewidth]{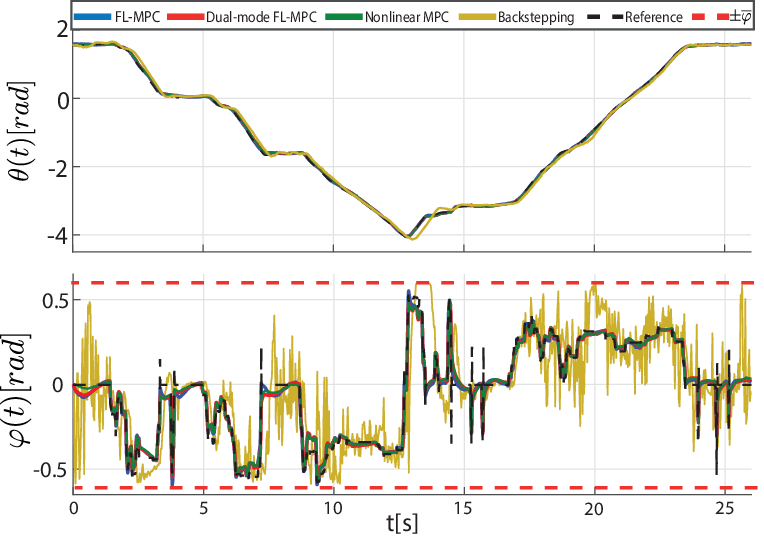}
	\caption{Experimental results: Heading and Steering Angles}
	\label{fig:results-angles}
\end{figure}
\begin{figure}[!h]
	\centering
	\includegraphics[width=1\linewidth]{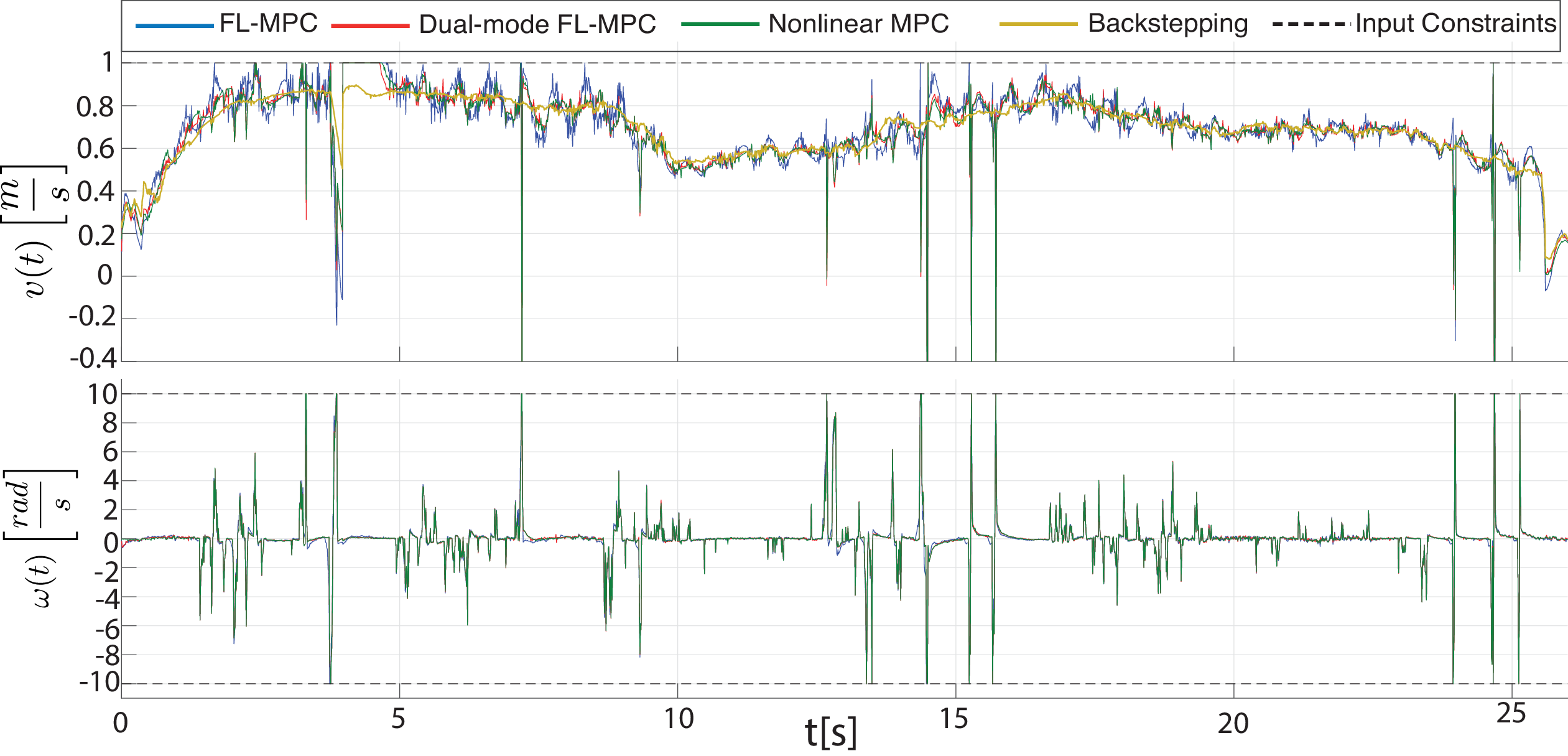}
	\caption{Experimental results: Control Inputs}
	\label{fig:results-inputs}
\end{figure}

\begin{figure}[!h]
	\centering
	\includegraphics[width=0.93\linewidth]{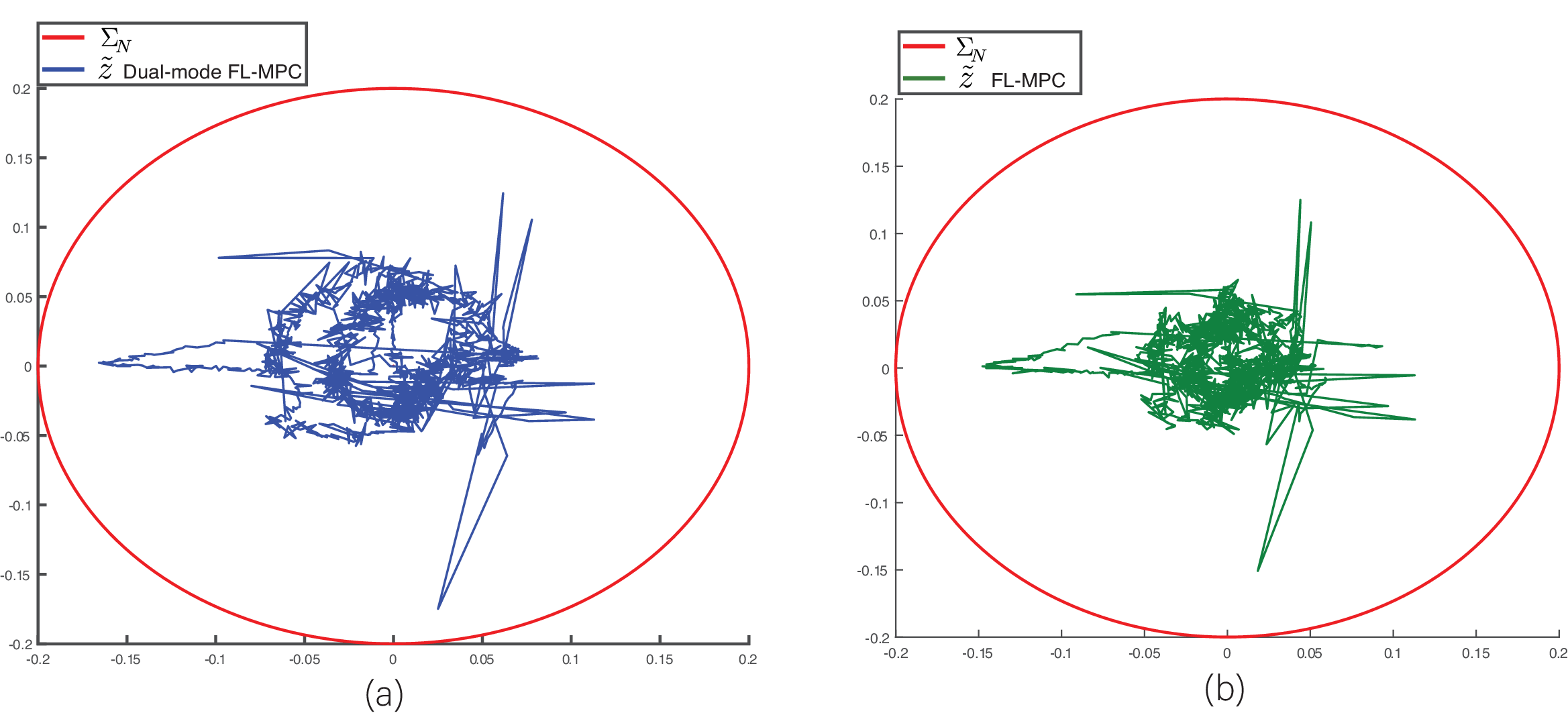}
	\caption{RPI set $\Sigma_N$ and $\tilde{z}(t):$  Dual-Mode FL-MPC (a) vs FL-MPC (b).}
	\label{fig:invariant-region}
\end{figure}
The obtained results are collected in Figs. \ref{fig:results-traj}-\ref{fig:results-inputs} and 
and Tables \ref{tab:results-1}-\ref{tab:results-3}. For the interested reader, videos of the performed experiments can be found at the following web link: 
\url{https://youtu.be/aeHZKyRfcEo}.
The tracking performance has been evaluated considering for a reference trajectory
that requires a maximum speed of $0.6\frac{m}{s}.$ It is worth mentioning that, for both the considered trajectory, several tests have been run, and the obtained results have been averaged in Table \ref{tab:results-1}.  

Fig. \ref{fig:results-traj} and Tables \ref{tab:results-1} show that the proposed controller achieves better tracking when compared to the Nonlinear MPC and Backstepping controllers.
In particular, the Nonlinear MPC showed poor tracking performance in all the performed experiments. 
This finds justification in the  nonconvex nature of optimization \eqref{eq:NL-MPC-optimization}, which converges to local minima and, consequently, to nonoptimal solutions.
Moreover, as shown in Table \ref{tab:results-3}, the reference tracking performance cannot be improved by increasing the prediction horizon. For example, for $N=10,$ the nonlinear optimization solver for \eqref{eq:NL-MPC-optimization} can take up to 
approximately $35 \,ms$ to obtain a solution, which is far above the considered sampling time $T_s=10 \,ms$.  On the other hand, the backstepping controller shows slightly better tracking performance of the Nonlinear MPC. However, a chattering phenomenon affects the computed steering angle, which is a control input for the approach in \cite{hu2021adaptive} (see the rapid and discontinuous switching of $\varphi(t)$ in Fig.~\ref{fig:results-angles}). Moreover, as shown in Fig. \ref{fig:results-inputs}, the other control input $v(t)$ computed by the backstepping controller is conservative, i.e. the longitudinal velocity never reaches the prescribed limits $\overline{v}$. The two above mentioned drawbacks justify why the tracking performance of the backstepping scheme are slightly worse of the one achieved with the proposed tracking controller degrades.

\begin{table}[!h]
\caption{Comparison of tracking performance: Trajectory 1}
\resizebox{\columnwidth}{!}{%
\begin{tabular}{|c|cc|cc|cc|}
\hline
\multirow{2}{*}{Algorithm} & \multicolumn{2}{c|}{Distance Error } & \multicolumn{2}{c|}{Heading Error} & \multicolumn{2}{c|}{Steering Error} \\ \cline{2-7} 
 & \multicolumn{1}{c|}{$ISE_{xy}$} & $ITSE_{xy}$& \multicolumn{1}{c|}{$ISE_{\theta}$} & $ITSE_{\theta}$ & \multicolumn{1}{c|}{$ISE_{\varphi}$} &$ITSE_{\varphi}$\\ \hline
FL-MPC & \multicolumn{1}{c|}{ 0.0301 } & 0.2207  & \multicolumn{1}{c|}{0.0187} & 0.1513 & \multicolumn{1}{c|}{0.0323} & 0.2971 \\ \hline
 Dual-mode FL-MPC & \multicolumn{1}{c|}{0.0571} & 0.4492 & \multicolumn{1}{c|}{0.0356} & 0.3204 & \multicolumn{1}{c|}{0.0351} &0.3576 \\ \hline
Nonlinear MPC & \multicolumn{1}{c|}{0.0585} & 0.4496 & \multicolumn{1}{c|}{0.0381 } &0.3287 & \multicolumn{1}{c|}{0.0381} &0.3666\\ \hline
Backstepping \cite{hu2021adaptive} & \multicolumn{1}{c|}{0.1754 } & 1.4428  & \multicolumn{1}{c|}{0.1889 } &  1.6460 & \multicolumn{1}{c|}{ 0.7842} &7.5517 \\ \hline
\end{tabular}
}
\label{tab:results-1}
\end{table}

On the other hand, the performance obtained with the FL-MPC is slightly superior to the ones obtained with the dual-mode MPC. 
The reason behind such a result is that  FL-MPC  uses the RPI region \eqref{eq:terminal-region-shaping} to ensure recursive feasibility, but it never directly relies on the associated controller, which is by nature more conservative. However, the dual-mode MPC implementation shows some computational advantage related to the fact that when the tracking error enters the terminal regions, then a simpler optimization problem \eqref{eq:terminal-optimization} is solved.
Moreover, the proposed FL-MPC controller fulfills the vehicle’s input constraints, as shown in Figs. \ref{fig:results-traj}-\ref{fig:results-inputs}.  In addition, Fig.~\ref{fig:invariant-region} shows the computed positively invariant regions $\Sigma_N$ and the trajectory of the linearized tracking error $\tilde{z}(t)$  for the two proposed strategies once the invariant region is entered (i.e., after 8-time steps), confirming the RPI nature of $\Sigma_N$.

\begin{table}[!h]
\caption{Comparison of average and maximum computational times (in ms) of Nonlinear and FL-MPC}
\resizebox{\columnwidth}{!}{%
\begin{tabular}{|c|cc|cc|cc|}
\hline
\multirow{2}{*}{Algorithm} & \multicolumn{2}{c|}{N=3} & \multicolumn{2}{c|}{N=5} & \multicolumn{2}{c|}{N=10} \\ \cline{2-7} 
 & \multicolumn{1}{c|}{MAX $(ms)$} & AVG$(ms)$ & \multicolumn{1}{c|}{MAX$(ms)$} & AVG$(ms)$ & \multicolumn{1}{c|}{MAX$(ms)$} &AVG$(ms)$\\ \hline
FL-MPC & \multicolumn{1}{c|}{ 3.4302} & 0.5344 & \multicolumn{1}{c|}{3.6239} & 0.7415 & \multicolumn{1}{c|}{3.4617} & 0.7865 \\ \hline
 Dual-mode FL-MPC & \multicolumn{1}{c|}{1.7245} & 0.3325 & \multicolumn{1}{c|}{3.4520} & 0.627 & \multicolumn{1}{c|}{2.2942} &0.3838 \\ \hline
Nonlinear MPC & \multicolumn{1}{c|}{6.2233} &  1.2933 & \multicolumn{1}{c|}{12.2231 } & 2.8243 & \multicolumn{1}{c|}{35.3515} &  5.3139\\ \hline
\end{tabular}
}
\label{tab:results-3}
\end{table}
Concerning the computational analysis, the proposed Dual-mode FL-MPC and FL-MPC solutions have been compared with the Nonlinear MPC  considering the prediction horizons $N=3, N=5, N=10$. The obtained results are collected in Table \ref{tab:results-3}. It can be appreciated how the computational times, both maximum and average ones, of the proposed solutions are significantly lower than the ones of the Nonlinear MPC, especially for larger values of $N.$ 

\section{Conclusions}\label{sec:conclusion}
In this paper, a novel Feedback Linearized Model Predictive Control strategy for input-constrained self-driving cars has
been presented. The proposed strategy combines two main ingredients: 1) an input-output FL technique and 2) a dual-mode MPC framework. The obtained tracking
controller has the peculiar capability of efficiently dealing with state-dependent input constraints acting on the feedback linearized car's model while ensuring recursive feasibility,
stability, and velocity constraints fulfillment. Extensive experimental results and comparisons have been carried out to highlight the features and advantages of the proposed tracking controller.






\bibliographystyle{IEEEtran}
\bibliography{IEEEabrv,references}

\vfill

\end{document}